\documentclass[runningheads,a4paper]{./llncs}

\usepackage{tikz}
\usetikzlibrary{shapes, positioning, arrows, automata, trees, arrows.meta}
\usepackage{amsmath}
\usepackage{amssymb}
\usepackage{wasysym} % sexe symbol
\usepackage{mathtools} % coloneqq
\usepackage{multirow}
\usepackage{bbold} % N, Z
\usepackage[ruled]{algorithm2e} % algo
\usepackage{bussproofs} % proof tree
\usepackage{cancel} % not a model
\usepackage{calc} % widthof in mylist
\usepackage{enumitem} % define mylist
\usepackage{times}

\makeatletter
\newcommand{\Lemma}[1]{Lemma~\ref{#1}}
\newcommand{\Theorem}[1]{Theorem~\ref{#1}}

\newcommand{\Proposition}[1]{Proposition~\ref{#1}}
\newcommand{\Algorithm}[1]{Algorithm~\ref{#1}}

\newcommand{\Section}[1]{Section~\ref{#1}}
\newcommand{\Example}[1]{Example~\ref{#1}}
\newcommand{\Appendix}[1]{Appendix~\ref{#1}}
\newcommand{\Remark}[1]{Remark~\ref{#1}}

\newcommand{\id}[1]{#1}

\renewcommand{\dot}{\ensuremath{.\;}}
\newcommand{\sem}[1]{\ensuremath{[\![#1]\!]}}

 \newcommand{\dom}{\text{dom}}
\newcommand{\ife}{\ensuremath{\triangleright}}
\newcommand{\bigife}{\scalebox{1.5}{$\triangleright$}}
\newcommand{\sync}{||} % {\ensuremath{\parallel}}
\newcommand{\Sync}{\text{Sync}}
\newcommand{\dec}[1]{\ensuremath{\text{dec}_{#1}}}
\newcommand{\s}{\circledast}
\newcommand{\splitsum}{\odot}

\newcommand{\N}{\mathbb{N}}
\renewcommand{\O}{\mathcal{O}}

\newcommand{\T}{\mathcal{T}}

\newcommand{\Z}{\mathbb{Z}}

\title{Decidable  Weighted  Expressions  with  Presburger  Combinators}

\author{
	Emmanuel Filiot \thanks{
		E. Filiot is a research associate of F.R.S.-FNRS. This work has been supported by the following projects: the ARC Project Transform (Federation Wallonie-Brussels), the FNRS CDR project Flare.
	},
	Nicolas Mazzocchi \thanks{
		N. Mazzocchi is a PhD funded by a FRIA fellowship from the F.R.S.-FNRS.
	}
	and Jean-Fran\c{c}ois Raskin \thanks{
		J.-F. Raskin is supported by an ERC Starting Grant (279499: inVEST), by the ARC project —- Non-Zero Sum Game Graphs: Applications to Reactive Synthesis and Beyond -— funded by the Fédération Wallonie-Bruxelles, and by a Professeur Francqui de Recherche grant awarded by the Francqui Fondation.
	}
}

\institute{Universit\'e libre de Bruxelles}

\begin{document}
	\maketitle

%% long abstract
	
	% \begin{abstract}
            % Weighted $(max,+)$-automata are automata which model
            % functions from finite words to integers. They have found
            % many applications, for instance in natural language
            % processing and more recently in computer-aided
            % verification. However, they suffer from several
            % limitations: they are not closed under operations such as
            % min and substraction, and basic decision problems such as
            % universality and comparison are undecidable. 

            % We introduce an expressive formalism, called synchronised
            % expressions with iterated sum, to define functions from
            % finite words to integers, by which to combine unambiguous
            % $(max,+)$-automata. Its expressiveness is incomparable
            % with $(max,+)$-automata, it is closed under any binary
            % operation definable in Presburger logic, and is strictly
            % more expressive than finite-valued $(max,+)$-automata. 
            % We show that emptiness,
            % universality and comparison are decidable problems for
            % synchronised expressions, and are
            % \textsc{PSpace-Complete} for the
            % fragment without iterated sum.
	% \end{abstract}

%% short abstract
	\begin{abstract}
            In this paper, we investigate the expressive power and the algorithmic properties of weighted expressions, which define functions from finite words to integers. First, we consider a slight extension of an expression formalism, introduced by Chatterjee. et. al. in the context of infinite words,  by which to combine values given by unambiguous $(\max,+)$-automata, using Presburger arithmetic. We show that important decision problems such as emptiness, universality and comparison are \textsc{PSpace-c} for these expressions. We then investigate the extension of these expressions with Kleene star. This allows to iterate an expression over smaller fragments of the input word, and to combine the results by taking their iterated sum. The decision problems turn out to be undecidable, but we introduce the decidable and still expressive class of synchronised expressions.
            %
            % Weighted $(max,+)$-automata are automata which model
            % functions from finite words to integers. They have found
            % many applications, for instance in natural language
            % processing and more recently in computer-aided
            % verification. However, they suffer from several
            % limitations: they are not closed under operations such as
            % min and substraction, and basic decision problems such as
            % universality and comparison are undecidable. 
            %
            % We introduce an expressive formalism, called synchronised
            % expressions with iterated sum, to define functions from
            % finite words to integers, by which to combine unambiguous
            % $(max,+)$-automata. Its expressiveness is incomparable
            % with $(max,+)$-automata, it is closed under any binary
            % operation definable in Presburger logic, and is strictly
            % more expressive than finite-valued $(max,+)$-automata. 
            % We show that emptiness,
            % universality and comparison are decidable problems for
            % synchronised expressions, and are
            % \textsc{PSpace-Complete} for the
            % fragment without iterated sum.
	\end{abstract}
	
	\section{Introduction}

\paragraph{Quantitative languages} Quantitative languages (QL), or series,
generalise Boolean languages to function from finite words into some
semiring. They have recently received
a particular attention from the verification community, for their
application in modeling system \emph{quality} \cite{cdh10}, lifting
classical Boolean verification problems to a quantitative
setting. In this paper, we consider the case of integer weights
and in this context, the comparison problem asks whether two
QL $f,g:\Sigma^*\rightarrow \mathbb{Z}$ satisfy
$f(u)\leq g(u)$ for all $u\in\Sigma^*$. Similarly, the universality ($f \geq \nu$
where $\nu$ is a constant) and equivalence problem ($f=g$) can be
defined, as well as emptiness (does there exists a word whose value is
above some given threshold). We say that a formalism for QL is decidable if all these problems are decidable. A popular formalism to define QL is that of
weighted automata (WA) \cite{Droste_Kuich_Vogler_2009}. However,
WA over the semiring $(\mathbb{Z}, \max, +)$, called
$(\max,+)$-automata, are undecidable \cite{Krob/94},
even if they are linearly ambiguous $(\max,+)$-automata
\cite{DBLP:journals/corr/DaviaudGM16}. 

\paragraph{Decidable formalisms for quantitative languages and objectives} The
largest known class of $(\max,+)$-automata enjoying decidability is that of
finitely ambiguous $(\max,+)$-automata, which is also expressively
equivalent to the class of finite-valued $(\max,+)$-automata (all the
accepting executions over the same input run yields a constant number
of different values) \cite{DBLP:conf/fsttcs/FiliotGR14}. Moreover,
$(\max,+)$-automata are not closed under simple operations such as $\min$ and
the difference $-$ \cite{DBLP:journals/tcs/KlimannLMP04}. Basic
functions such as $u\mapsto \min(\#_a(u),\#_b(u))$
and\footnote{$\#_\sigma(u)$ is the number of occurrences of $\sigma$
  in $u$} (as a consequence) $u\mapsto |f(u)-g(u)|$ are
not definable by $(\max,+)$-automata, even if $f,g$ are
\cite{DBLP:journals/tcs/KlimannLMP04}.To cope with the
expressivity and undecidability issues, a class of weighted expressions
was introduced in \cite{DBLP:conf/concur/ChatterjeeDEHR10} in the
context of $\omega$-words. Casted to finite words, the idea
is to use deterministic $(\max,+)$-automata as atoms, and to
combine them using the operations $\max$, $\min$, $+$, and $-$.  The decision
problems defined before were shown to be \textsc{PSPace-c}
\cite{DBLP:conf/icalp/Velner12} over $\omega$-words. One limitation of
this formalism, casted to finite words, if that it is not expressive
enough to capture finitely ambiguous $(max,+)$-automata, yielding two
incomparable classes of QL. In this paper, our
objective is to push the expressiveness of weighted expressions as far
as possible while retaining decidability, and to capture both finitely
ambiguous $(\max,+)$-automata and the expressions of
\cite{DBLP:conf/concur/ChatterjeeDEHR10}, for finite words.

\paragraph{Monolithic expressions with Presburger combinators}
We define in \Section{sec:mono} a class of expressions, inspired from
\cite{DBLP:conf/concur/ChatterjeeDEHR10}, that we call monolithic in
contrast to another class of expressions defined in a second
contribution. The idea is to use unambiguous $(\max,+)$-automata as
atoms, and to combine them using $n$-ary functions definable in
Presburger arithmetics (we call them Presburger combinators). Any finitely ambiguous $(\max,+)$-automaton
being equivalent to a finite union of unambiguous ones
\cite{DBLP:conf/fsttcs/FiliotGR14}, this formalism captures finitely
ambiguous $(\max,+)$-automata (using the Presburger combinator $\max)$. 
We show that all the decision problems are \textsc{PSpace-c}, matching
the complexity of \cite{DBLP:conf/icalp/Velner12}. It is important to mention that this
complexity result cannot be directly obtained from
\cite{DBLP:conf/icalp/Velner12} which is on $\omega$-words with
mean-payoff automata as atoms (hence the value of an infinite word is
prefix-independent). Moreover, unlike in
\cite{DBLP:conf/icalp/Velner12}, we can rely on existing results by encoding
expressions into reversal-bounded counter machines \cite{ibarra}.

\paragraph{Expressions with iterated sum} The previous expressions are
monolithic in the sense that first, some values are computed by weighted
automata applied on the whole input word, and then these values are
combined using Presburger combinators. It is not possible to iterate
expressions on factors of the input word, and to aggregate all the
values computed on these factors, for instance by a sum operation. The
basic operator for iteration is that of Kleene star (extended to
quantitative languages), which we call more explicitly \emph{iterated
  sum}. It has already been defined
in \cite{Droste_Kuich_Vogler_2009}, and its unambiguous version
considered in \cite{DBLP:conf/csl/AlurFR14} to obtain an expression
formalism equivalent to unambiguous $(\max,+)$-automata. 
 Inspired by \cite{DBLP:conf/csl/AlurFR14},  we investigate in
 \Section{sec:iter} the extension of monolithic expressions
 with unambiguous iterated sum, which we just call iterated sum in the
 paper. The idea is as follows: given an
 expression $E$ which applies on a domain $D$, the expression $E^\s$ is
 defined only on words $u$ that can be uniquely decomposed (hence the name
 unambiguous)  into factors $u_1u_2\dots u_n = u$ such that $u_i\in D$, and the value of $u$ is
 then $\sum_{i=1}^n E(u)$. Unfortunately, we show that such an
 extension yields undecidability (if 2 or more iterated sum
 operations occur in the expression). The undecidability is caused by
 the fact that subexpressions $E^\s$ may decompose the input word in
 different ways. We therefore define the class of so called
 \emph{synchronised} expressions with iterated sum, which forbids this
 behaviour. We show that while being expressive (for instance, they
 can define QL beyond finitely ambiguous
 $(\max,+)$-automata), decidability is recovered. The proof goes via a
 new weighted automata model (\Section{sec:chop}), called \emph{weighted chop automata},
 that slice the input word into smaller factors, recursively apply
 smaller chop automata on the factors to compute their values, which
 are then aggregated by taking their sum. In their synchronised
 version, we show decidability for chop automata. We finally discuss
 some extensions in \Section{sec:conclu}\footnote{Due to lack of
   space, full proofs are given in Appendix.}.

% \paragraph{Organisation of the paper} Preliminary notions are
% introduced in \Section{sec:prelim}. The monolithic expressions are studied
% in \Section{sec:mono}. Their extension with unambiguous iterated sum is
% presented in \Section{sec:iter}, whose decidability via weighted automata
% is given in \Section{sec:chop}. Finally,  \Section{sec:conclu}
% discusses some extensions and a comparison with \cite{DBLP:conf/csl/AlurFR14}. 

%%% Local Variables:
%%% mode: latex
%%% TeX-master: t
%%% End:

	\section{Quantitative Languages}\label{sec:prelim}

\paragraph{Words,  languages and quantitative languages}
Let $\Sigma$ be a finite alphabet
and denote by $\Sigma^*$ the set of finite words over $\Sigma$, with
$\epsilon$ the empty word.  Given two words $u,v \in \Sigma^*$, $|u|$
and $|v|$ denote their length, and the
distance between $u$ and $v$ is defined as $d(u,v)=|u|+|v|-2
|\sqcap(u,v)|$, where $\sqcap(u,v)$ denotes the longest common prefix
of $u$ and $v$. A \emph{quantitative language} (QL)\footnote{Also called \emph{formal series} in \cite{Droste_Kuich_Vogler_2009}} is a partial function $f : \Sigma^*\rightarrow \Z$, whose domain is denoted by $\dom(f)$.
E.g., consider the function mapping any word $w\in\Sigma^*$ to the number of
occurrences $\#_\sigma(w)$ of some symbol $\sigma\in\Sigma$ in $w$. A QL $f$ is \emph{Lipschitz-continuous} if there
exists $K \in \N$ such that for all words $u,v \in \Sigma^*$, $|f(u) -
f(v) | \leq K \cdot d(u,v)$.

\paragraph{Combinators for quantitative languages}
Any binary operation $\boxplus : \Z^2\rightarrow \Z$ is extended to quantitative languages by $f_1\boxplus f_2(w) = f_1(w)\boxplus f_2(w)$ if $w\in\dom(f_1)\cap \dom(f_2)$, otherwise it is undefined.
We will consider operations defined in \emph{existential Presburger logic}.
An existential Presburger formula (simply called
Presburger formula in the sequel) is built over terms $t$ on the signature $\{0,1,+\}\cup X$, where $X$ is a set of variables, as follows: $\phi\ \Coloneqq \ t=t\mid t>t\mid \phi\vee\phi\mid \phi\wedge\phi\mid \exists x\dot \phi$.
If a formula $\phi$ has $n+1$ free variables $x_1,\dots,x_{n+1}$, for all $v_1,\dots,v_{n+1}\in\Z$, we write $\phi(v_1,\dots,v_{n+1})$ if $\phi$ holds for the valuation mapping $x_i$ to $v_i$.
When $n\geq 1$, we say that $\phi$ is \emph{functional} if for all $v_1,\dots,v_{n}\in\Z$, there exists a unique $v_{n+1}\in\Z$ such that $\phi(v_1,\dots,v_{n+1})$ holds.
Hence, $\phi$ defines a (total) function from $\Z^{n}$ to $\Z$ that we denote $\sem{\phi}$.
We call $n$ the arity of $\phi$ and may write $\phi(x_1,\dots,x_n)$ to denote the unique $x_{n+1}$ such that $\phi(x_1,\dots,x_{n+1})$ holds.
We say that a function $f : \Z^n \rightarrow \Z$ is Presburger-definable if there exists a functional Presburger-formula $\phi$ such that $f = \sem{\phi}$.
E.g., the $\max$ of values $x_1,\dots,x_n$ is definable by
$\phi_{\max}(x_1,\dots,x_n,x) \equiv (\bigwedge_{i=1}^n x_i \leq x)\wedge (\bigvee_{i=1}^n x_i = x)$.

\paragraph{Semi-linear sets} Let $k\geq 1$. A set $S \subseteq
\mathbb{Z}^k$ is \emph{linear} if there exist $x_1,\dots,x_n\in \Z^k$,
called the period vectors, and $x_0\in \Z^k$, called the base, such
that $S = \{ x_0+\sum_{i=1}^n a_ix_i\mid a_1,\dots,a_n\in \N\}$. 
$S$ is \emph{semi-linear} if it is a finite a union of linear
sets. Note that the set of base and periodic vectors of each linear
set of the union provides a finite representation of $S$. It is a
folklore result that a set $S\subseteq \Z^k$ is semi-linear iff it is
definable by some existential Presburger formula.

\paragraph{Decision problems}
In this paper, we are interested by fundamental decision problems on (finite representations of) quantitative languages, namely universality, emptiness and comparison.
Given finitely represented quantitative languages $f, f_1, f_2$ and $v\in\Z$,
\begin{itemize}
\item
	the $v$-emptiness (resp. $v$-universality) problem asks whether there exists  $u\in\dom(f)$ such that $f(u)\succsim v$ (resp. whether all $u\in\dom(f)$ satisfies $f(u)\succsim v$), for $\succsim\ \in\{>,\geq\}$. 
\item
	the $\succsim$-inclusion problem (denoted $f_1\succsim f_2$) with $\succsim\ \in\{>,\geq\}$ asks whether $\dom(f_1)\supseteq \dom(f_2)$ and for all $w\in\dom(f_2)$, $f_1(w)\succsim f_2(w)$. 
\item
	the equivalence problem, denoted $f_1\equiv f_2$, asks whether $f_1{\geq} f_2\wedge f_2{\geq} f_1$. 
\end{itemize}

\begin{remark} \label{remark:zeroEmptiness}
For classes of QL  (effectively) closed under regular
domain restriction and difference, and with decidable domain
inclusion,  the $v$-universality, inclusion and equivalence
problems, are reducible to the $0$-emptiness problem as follows:
  \begin{enumerate}
  	\item
	  	to establish $\forall w \in \dom(f): f(w) \geq v$ (universality), it suffices to check that it is not the case that $\exists w \in \dom(f): -(f(w)-v) > 0$ ($0$-emptiness).
	\item
		to establish $\dom(f_2)\subseteq \dom(f_1)$ and for all $w\in\dom(f_2)$, $f_1(w)\geq f_2(w)$, when the first check succeeds, we reduce the second one as follows: construct a new QL $g$ on $\dom(f_2)$ such that $\forall w\in\dom(f_2): g(w)=f_2(w)-f_1(w)$ and check that  $\forall w \in \dom(f_2): g(w) \geq 0$ ($0$-emptiness).

	% \item to establish $\dom(V_2)\subseteq \dom(V_1)$ and for all $w\in\dom(V_2)$, $V_1(w)\geq V_2(w)$, when the first check succeed, we reduce the second check to emptiness as follows: we construct a new quantitative language $V_3$ on $\dom(V_2)$ such that $\forall w\in\dom(V_2): V_3(w)=V_2(w)-V_1(w)$ and check that  $\forall w \in \dom(V_2): V_3(w) \geq 0$.
  \end{enumerate}
  \noindent
The other variants with strict inequalities are treated
similarly. Note also with similar arguments, we can show that the
$0$-emptiness problem can be reduced to the universality and the
inclusion problems. The quantitative expression formalisms that we
define in this paper have those closure properties (in \textsc{PTime}) and so, we
concentrate, in most of our results, on the $0$-emptiness problem.
\end{remark}

\paragraph{Weighted automata}
Weighted automata (WA) have been defined as a representation of QL (more generally with values in a semiring).
Here, we consider weighted automata over the semiring $(\Z\cup \{-\infty\},\max,+)$ and just call them \emph{weighted automata}.
They are defined as tuples $M = (A, \lambda)$ where $A = (Q,I,F,\Delta)$ is a finite automaton over $\Sigma$ whose language is denoted by $L(A)$ and $\lambda : \Delta\rightarrow \Z$ is a weight function on transitions.
Given a word $w\in L(A)$ and an accepting run $r = q_1a_1\dots
q_na_nq_{n+1}$ of $A$ on $w$, the value $V(r)$ of $r$ is defined by
$\sum_{i=1}^n \lambda(q_i,a_i,q_{i+1})$ if $n>1$, and by $0$
if\footnote{Sometimes, initial and final weight functions are
  considered in the literature \cite{Droste_Kuich_Vogler_2009}, so that non-zero values can be assigned to $\epsilon$} $n=1$.
Finally, $M$ defines a quantitative language $\sem{M} :
L(A)\rightarrow \Z$ such that for all $w\in L(A)$, $\sem{M}(w) =
\max\{ V(r)\mid r\text{ is an accepting run of $A$ on $w$}\}$.
$M$ is called deterministic if $A$ is deterministic. We say that $M$ is $k$-ambiguous if $A$ is $k$-ambiguous, i.e. there are at most $k$ accepting runs on words of $L(A)$.
A $1$-ambiguous WA is also called \emph{unambiguous}.
$M$ is $k$-valued if for all $w\in L(A)$, the set $\{ V(r)\mid r\text{ is an accepting run of $A$ on $w$}\}$ has cardinality at most $k$.
In particular, any $k$-ambiguous WA is $k$-valued.
The converse also holds, and it is decidable whether a WA is $k$-valued, for a given $k$ \cite{DBLP:conf/fsttcs/FiliotGR14}.
While emptiness is decidable for WA~\cite{DBLP:journals/corr/abs-1111-0862}, inclusion and universality are undecidable~\cite{Krob/94}.
However, all these problems are decidable for $k$-valued WA, for a fixed $k$~\cite{DBLP:conf/fsttcs/FiliotGR14}.

	\section{Monolithic Expressions}\label{sec:mono}
\newcommand{\drawLastBlock}{
	\begin{tikzpicture}[>=stealth, node distance=3cm, ultra thick]
		\renewcommand{\id}[1]{}
		
		% ~~~~~~~~~~~~~~~~~~~~~~~~~~~~~ STYLES ~~~~~~~~~~~~~~~~~~~~~~~~~~~~~ %
		
		\tikzstyle{accepting}=[
		accepting by arrow,
		accepting text=,
		accepting where= right
		]
		
		\tikzstyle{initial}=[
		initial by arrow,
		initial text=LastBlock,
		initial where= left
		]
		
		% ~~~~~~~~~~~~~~~~~~~~~~~~~~~~~ NODES ~~~~~~~~~~~~~~~~~~~~~~~~~~~~~ %
		
		\node[state, initial] (0) {\id{0}};
		\node[state, accepting] (1) [right of = 0] {\id{1}};
		
		% ~~~~~~~~~~~~~~~~~~~~~~~~~~~~~~ ARCS ~~~~~~~~~~~~~~~~~~~~~~~~~~~~~~ %
		
		\path[->]
			(0) edge [loop above] node [above] {$\begin{array}{l|l|l} a & 0 \quad \$ & 0 \end{array}$} (0)
			(0) edge node [above] {$\begin{array}{l|l} a & 1 \end{array}$} (1)
			(1) edge [loop above] node [above] {$\begin{array}{l|l} a & 1 \end{array}$} (1)
		;
	\end{tikzpicture}
}

We start our study of weighted expressions by a definition directly
inspired by~\cite{DBLP:conf/concur/ChatterjeeDEHR10} where weighted
automata\footnote{Chatterjee et al. studied quantitative expressions
  on infinite words and the automata that they consider are
  deterministic mean-payoff automata.} are used as building blocs of
quantitative expressions that can be inductively composed with
functions such as $\min$, $\max$, addition and difference. The
equivalence checking problem for those expressions is decidable in
\textsc{PSpace}. We start here with deterministic $(\max,+)$-automata as
building blocs.

\begin{definition}
A \emph{simple expression (s-expression)} is a term $E$ generated by
    $$
    E\ \Coloneqq \ D\mid \min(E_1,E_2) \mid \max(E_1,E_2) \mid E_1+E_2 \mid E_1-E_2
    $$
 \noindent
 where $D$ is a deterministic WA (we remind that by WA we mean $(\max,+)$-automata).
\end{definition}

\noindent
Any  s-expression $E$ defines a quantitative language $\sem{E} : \Sigma^*\rightarrow \mathbb{Z}$ on a domain $\dom(E)$ inductively as follows:
if $E \equiv A$, then $\dom(E) = L(A)$ and for all 
$u\in L(A)$, $\sem{E}(u) = \sem{A}(u)$ (the semantics of WA is defined in \Section{sec:prelim}); if
$E \equiv \min(E_1,E_2)$, then $\dom(E) = \dom(E_1) \cap \dom(E_2)$ and
for all $u\in\dom(E)$, $\sem{E}(u) =
\min(\sem{E_1}(u),\sem{E_2}(u))$, symmetrical works for $\max$, $+$ and $-$. 
We say that two s-expressions $E_1,E_2$ are equivalent if
$\sem{E_1} = \sem{E_2}$ (in particular $\dom(E_1)=\dom(E_2)$). 
To characterise the expressiveness of s-expressions, we note that: 

\begin{lemma}\label{lem:s-exprexpr}
Any s-expression defines a Lipschitz continous quantitative language.
\end{lemma}
 
As we show in Appendix, unambiguous WA can define non Lipschitz
continuous functions, hence not definable by s-expressions. On the contrary, the function $u\mapsto \min(\#_a(u),\#_b(u))$ is definable by an s-expression while it is not definable by a WA~\cite{DBLP:journals/tcs/KlimannLMP04}.

\begin{proposition}\label{prop:s-exprbis}
There are quantitative languages that are definable by unambiguous
weighted automata and not by s-expressions. There are quantitative languages that
are definable by s-expressions but not by a WA.
\end{proposition}

To unleash their expressive power, we generalise s-expressions. First, instead deterministic WA, we consider unambiguous WA as atoms. This extends their expressiveness beyond finite valued WA. Second, instead of considering a fixed (and arbitrary) set of composition functions, we consider any function that is (existential) Presburger definable. Third, we consider the addition of Kleene star operator. While the first two extensions maintain decidability in \textsc{PSpace}, the third extension leads to undecidability and sub-cases need to be studied to recover decidability. We study the two first extensions here and the Kleene star operator in the next section.

\begin{definition}
    \emph{Monolithic expressions (m-expression)} are terms $E$
    generated by the grammar $
    E\ \Coloneqq \ A\mid \phi(E_1,\dots,E_n)
    $,     where $A$ is an unambiguous WA, and $\phi$ is a
    functional Presburger formula of arity $n$. 
\end{definition}
\noindent
The semantics $\sem{E} : \Sigma^*\rightarrow \mathbb{Z}$ of an m-expression $E$
is defined inductively, and similarly as s-expression. In particular,
for $E = \phi(E_1,\dots,E_n)$, $\dom(E) = \bigcap_{i=1}^n\dom(E_i)$
and for all $u\in\dom(E)$, $\sem{E}(u) =
\sem{\phi}(\sem{E_1}(u),\dots,\sem{E_n}(u))$ (the semantics of functional
Presburger formulas is defined in \Section{sec:prelim}).

% As for s-expression,  any  m-expression $E$ defines a quantitative language $\sem{E} : \Sigma^*\rightarrow \mathbb{Z}$ on a domain $\dom(E)$ inductively as follows:
% if $E \equiv A$, then $\dom(E) = L(A)$ and for all 
% $u\in L(A)$, $\sem{E}(u) = \sem{A}(u)$ (the semantics of WA is defined in \Section{sec:prelim}); if
% $E \equiv \phi(E_1,\dots,E_n)$, then $\dom(E) = \bigcap_{i=1}^n\dom(E_i)$ and
% for all $u\in\dom(E)$, $\sem{E}(u) =
% \sem{\phi}(\sem{E_1}(u),\dots,\sem{E_n}(u))$ (the semantics of functional
% Presburger formulas is defined in \Section{sec:prelim}).
% We say that two m-expressions $E_1,E_2$ are equivalent if
% $\sem{E_1} = \sem{E_2}$ (in particular $\dom(E_1)=\dom(E_2)$). 

\begin{example}
    As seen in \Section{sec:prelim}, $\max$ is
    Presburger-definable by a formula $\phi_{\max}$, it is also the case for $\min(E_1,\dots,E_n)$,
    $E_1+E_2$, $E_1-E_2$ and the unary operation $-E$.
    For m-expressions $E_1,E_2$, the
    distance $|E_1-E_2|\ :\ w\in\dom(E_1)\cap\dom(E_2)\mapsto
    |E_1(w)-E_2(w)|$ is defined by the m-expression 
    $\max(E_1-E_2,E_2-E_1)$. This function is not definable by a
    WA even if $E_1,E_2$ are 2-ambiguous WA, as a consequence of the
    non-expressibility by WA of $\min(\#_a(.),\#_b(.)) = |0-\max(-\#_a(.),-\#_b(.))|$
    \cite{DBLP:journals/tcs/KlimannLMP04}.
\end{example}

%\begin{proposition}\label{prop:expressivitymono}
%    The following hold:
%    \begin{enumerate}
%      \item There exists an unambiguous WA $A$ which is not definable by any
%    deterministic m-expression. 
%  \item For all finite-valued WA $B$, there
%    exists an m-expression $E$ such
%    that $\sem{B} = \sem{E}$. The converse does not hold. 
%\end{enumerate}
%\end{proposition}

\begin{lemma}\label{lem:m-exprexpr}
M-expressions are more expressive than finite valued WA. There are functions definable by m-expressions and not by a WA.
\end{lemma}

We finally come to the main result of this section:

\begin{theorem} \label{thm:mono}
    For m-expressions, the emptiness, universality and comparison
    problems
    are \textsc{PSpace-Complete}. 
\end{theorem}

\begin{proof}[Sketch]
    By \Remark{remark:zeroEmptiness}, all the problems
    reduce in \textsc{PTime} to the $0$-emptiness problem for which we establish \textsc{PSpace} membership. 
    Clearly, by combining Presburger formulas, any
    m-expression is equivalent to an m-expression
    $\phi(A_1,\dots,A_n)$ where $A_i$ are unambiguous WA. Now, the
    main idea is to construct a product $A_1\times \dots \times A_n$
    (valued over $\mathbb{Z}^n$), which maps any word $u\in \bigcap_i
    \dom(A_i)$ to $(A_1(u),\dots,A_n(u))$. (Effective) semi-linearity of 
    $\text{range}(A_1\times \dots \times A_n)$ is a consequence
    of Parikh's theorem, which implies semi-linearity of  
    $\text{range}(\phi(A_1,\dots,A_n))$. Then
    it suffices to check for the existence of a positive value in
    this set. To obtain \textsc{PSpace} complexity, the
    difficulty is that $A_1\times \dots \times A_n$ has exponential
    size. To overcome this, we encode $\phi(A_1,\dots,A_n)$ into a counter
    machine. First, $A_1\times \dots \times A_n$ is encoded into a
    machine $M$ whose counter valuation, after
    reading $u$, encodes the tuple $(A_1(u),\dots,A_n(u))$. Then,
    $M$ is composed with another counter machine $M_\phi$ that
    compute, on reading the word $\epsilon$, the value
    $\phi((A_1(u),\dots,A_n(u))$ (stored in an extra
    counter). Finally, the compositional machine $M\cdot M_\phi$ accepts iff this latter value
    is positive, hence it suffices to check for its emptiness. We
    define $M\cdot M_\phi$ in such a way that it is
    \emph{reversal-bounded} (its counters change from increasing to
    decreasing mode a constant number of times
    \cite{ibarra}). Reversal-bounded counter machines have decidable
    emptiness problem. While $M_\phi$ can be constructed in \textsc{PTime}, $M$
    has an exponential size in general. However, we can use a small
    witness property given in \cite{ibarra} to devise a \textsc{Pspace}
    algorithm that does not construct $M$ explicitly.

\textsc{PSpace-hardness} for emptiness is obtained from the emptiness problem
of the intersection of $n$ DFAs. 
    \qed
\end{proof}

	\section{Expressions with iterated sum}\label{sec:iter}

% \subsection{Definition and Undecidability}

Given $f : \Sigma^*\rightarrow \mathbb{Z}$ a quantitative language, the iterated sum of $f$ (or unambiguous Kleene star), denoted by $f^\s$, is defined by $f^\s(\epsilon) = 0$, and for all $u\in \Sigma^+$, if there exists at most one tuple $(u_1,\dots,u_n)\in (\dom(f)\setminus\{\epsilon\})^n$ such that $u_1\dots u_n = u$, then $f^\s(u) = \sum_{i=1}^n f(u_i)$.
Note that $\epsilon\in \dom(f^\s)$ for any $f$.
By extending m-expressions with iterated sum, we obtain iterated-sum expressions (i-expressions).

\begin{definition}
    An iterated-sum expression $E$ (i-expression for short) is a term generated by the
    grammar $E\ \Coloneqq\ A\mid \phi(E,E) \mid E^\s$, 
    where $A$ is some unambiguous WA over $\Sigma$ and $\phi$ is a
    functional Presburger formula. 
\end{definition}

As for m-expressions, the semantics of any i-expression $E$ is a
quantitative language $\sem{E}:\Sigma^*\rightarrow \mathbb{Z}$ 
inductively defined on the structure of the expression.

% \myparagraph{Semantics} Any iter-expression $E$ defined a quantitative
% language $\sem{E}:\Sigma^*\rightarrow \mathbb{Z}$ on a domain
% $\dom(E)$ inductively as follows:
% \begin{itemize}
  % \item $E\equiv L/v$: then $\dom(E) = L$ and for all $u\in L$,
    % $\sem{L/v}(u) = v$,
  % \item $E \equiv E_1\ife E_2$: then $\dom(E) = \dom(E_1)\cup
    % \dom(E_2)$ and 
    % $$
    % \begin{array}{rrcl}
    % \sem{E} : & \dom(E) & \rightarrow & \mathbb{Z} \\
    % & u & \mapsto &
    % \begin{cases}
	    % \hfil\sem{E_1}(u) & \caseif u\in\dom(E_1)\\
	    % \hfil\sem{E_2}(u) & \caseif u\in\dom(E_2)\setminus\dom(E_1)
	% \end{cases}
	% \end{array}
    % $$
  % \item $E \equiv E_1\oplus E_2$, then $\dom(E)$ is the set of words
    % $u\in\Sigma^*$ such that $u$ can be \emph{uniquely} decomposed
    % into $u_1u_2$, where $u_i\in\dom(E_i)$. Then,
    % $\sem{E}(u) = \sem{E_1}(u_1)+\sem{E_2}(u_2)$. 
  % \item $E\equiv \phi(E_1,E_2)$, then the semantics is defined
    % similarly as for flat expressions.
  % \item $E = F^*$, then $\dom(E)$ is the set of words $u\in\Sigma^*$ which can be
    % uniquely decomposed into $u_1\dots u_n$ where\footnote{Note that $\dom(E) = \varnothing$
      % whenever $\epsilon\in\dom(F)$}
    % $u_i\in\dom(F)$. Then $\sem{E}(u) = \sum_{i=1}^n \sem{F}(u_i)$. 
% \end{itemize}

\begin{example}\label{ex:iterexpr}
    Assume that $\Sigma = \{ a, b, \$\}$ and consider the QL $f$ defined for all $u\in\Sigma^*$ by $ u_1\$ u_2 \$ \dots u_n\$\ \mapsto\ \sum_{i=1}^n \max(\#_a(u_i),\#_b(u_i)) $ where each $u_i$ belongs to $\{a,b\}^*$, and $\#_\sigma$ counts the number of occurrences of $\sigma$ in a word.
    Counting the number of $\sigma$ in $v\$$ where $v\in\{a,b\}^*$ is realisable by a 2 states deterministic WA $A_\sigma$.
    Then, $f$ is defined by the i-expression $\max(A_a, A_b)^\s$.
\end{example}

We show a positive and a negative result.
\begin{proposition}\label{prop:reg}
    The domain of any i-expression is (effectively) regular. 
\end{proposition}

% We say that a functional Presburger formula $\phi$ defines addition if
% for all $x,y\in\mathbb{Z}$, $\sem{\phi}(x,y)=x+y$.
% \red{En general dans ce papier, les Pers formulas sont restraintes a une unique operation ???}

% \begin{proposition}[Alur] \label{prop:alur}
    % Let $f$ be a quantitative language. The following are equivalent:
    % \begin{enumerate}
      % \item $f$ is definable by some Presburger-free i-expression,
      % \item $f$ is definable by some unambiguous sum-automaton. 
    % \end{enumerate}
% \end{proposition}

% \begin{proof}
    % PROOF: same construction as in the proof of the more general
    % result hybrid to iter. 
% \end{proof}

% \begin{remark}
    % As an immediate corollary of \Proposition{prop:alur} and
    % \Theorem{thm:mono}, the following expressions (with their
    % intuitive semantics) have decidable  emptiness, universality and
    % comparison problems: $E\ ::=\ F\mid \phi(E,E)\qquad F\ ::=\ L/v\mid F\oplus F\mid F\ife
    % F\mid F^*$.

    % Iterated-sum expressions are a flattening of this definition, in
    % the sense that we mix all operators together. Unfortunately, as
    % shown by the following theorem, this extension turns out to be
    % undecidable. 
% \end{remark}

\begin{theorem}\label{thm:undec}
    Emptiness, universality and comparisons for i-expressions are
    undecidable problems, even if only s-expressions are iterated.
    %Presburger formulas are restricted to
    %define min or max, and star operations are not nested. 
\end{theorem}

\begin{proof}[Sketch]
The proof of this theorem, inspired by the proof of
\cite{DBLP:journals/corr/DaviaudGM16} for the undecidability of WA
universality, consists of a reduction from the 2-counter machine halting problem to the $0$-emptiness problem of i-expressions. 
This establishes undecidability for the other decision problems by
\Remark{remark:zeroEmptiness}. In this reduction, a transition between
two successive configurations $...(q_1,(x \mapsto c_1, y \mapsto d_1))
\delta (q_2,(x \mapsto c_2, y \mapsto d_2))...$ is coded by a factor
of word of the form: $ ... \vdash q_1 a^{c_1} b^{d_1} \triangleleft
\delta \triangleright q_2 a^{c_2} b^{d_2} \dashv \vdash q_2 a^{c_2}
b^{d_2} \triangleleft ...$. 

We show that such a word encodes an halting computation if it respects
a list of simple requirements that are all are regular but two: one that expresses that increments
and decrements of variables are correctly executed, and one that
imposes that, from one transition encoding to the next, the current
configuration is copied correctly. In our example above, under the
hypothesis that $x$ is incremented in $\delta$, this amounts to
check that the number of $a$ occurrences before $\delta$ is equal to
the number of occurrences of $a$ after $\delta$ minus one. This
property can be verified by s-expression 
on the factor between the $\vdash$ and $\dashv$ that returns $0$ if
it is the case and a negative value otherwise. The second property
amounts to check that the number of occurrences of $a$ between the
first $\triangleright$ and $\dashv$ and the number of $a$ between the
second $\vdash$ and second $\triangleleft$ are equal. Again, it is
easy to see that this can be done with an s-expression that returns $0$ if it is the case and a negative value
otherwise. Then, with i-expressions we decompose the word into factors
that are between the markers $\vdash$ and $\dashv$, and other factors that are
between the markers $\triangleright$ and $\triangleleft$, and we iterate the
application of the s-expressions mentioned above. The sum of all the
values computed on the factors is equal to $0$ if the requirements are met and negative otherwise.\qed
\end{proof}

A close inspection of the proof above, reveals that the undecidability stems
from the asynchronicity between parallel star operators, and in the way
they decompose the input word (decomposition based on  $\vdash\dots
\dashv$ or $\triangleright\dots\triangleleft$). The two overlapping decompositions are needed. 
By disallowing this, decidability is recovered: subexpressions $F^\s$ and $G^\s$
at the same nested star depth must decompose words in exactly the same way. 

Let us formalise the notion of star depth. Given an
i-expression $E$, its syntax tree $T(E)$ is a tree 
labeled by functional Presburger formulas $\phi$, star operators
$^\s$, or unambiguous WA $A$. Any node $p$ of $T(E)$ defines a
subexpression $E|_p$ of $E$. %When a subexpression $E|_p$ occurs only
%once in $E$, we identify $E|_p$ and $p$. %
The \emph{star depth} of node $p$
is the number of star operators occurring above it, i.e. the number of
nodes $q$ on the path from the root of $T(E)$ to $p$ (excluded)
labeled by a star operator. E.g. in the expression $\phi(A_1^\s,
\phi(A_2^\s))^\s$, the subexpression $A_1^\s$ has star depth $1$, $A_1$
has star depth $2$, and the whole expression has star depth $0$.

% Let $L\subseteq \Sigma^*$. We denote by $L^\#$ the set of words
% $u\in\Sigma^*$ such that $u$ is uniquely decomposed into $u_1\dots
% u_n$ such that $u_i\in L$ for all $i$. Note that $\dom(E^*) =
% \dom(E)^\#$ for all i-expression $E^*$. We denote by $\dec{L}$ the
% (partial) function mapping all $u\in L^\#$ to its decomposition
% $(u_1,\dots,u_n)$. 

\begin{definition}
    An $i$-expression $E$ is \emph{synchronised} if for all nodes
    $p,q$ of $T(E)$ at the same star depth, if $E|_p = F^\s$ and $E|_q =
    G^\s$, then $\dom(F) = \dom(G)$. 
\end{definition}

By \Proposition{prop:reg}, this property is decidable. 
Asking that $F$ and $G$ have the same domain enforces that any word
$u$ is decomposed in the same way by $F^\s$ and $G^\s$. 
Given a set $S = \{E_1,\dots,E_n\}$ of $i$-expressions, we write
$\Sync(S)$ the predicate which holds true iff
$\phi(E_1,\dots,E_n)$ is synchronised, where $\phi$ is some arbitrary
functional Presburger formula of arity $n$. 

\begin{example}
An i-expression $E$ is star-chain if for any distincts subexpressions $F^\s$ and $G^\s$ of $E$, $F^\s$ is a subexpression
of $G$, or $G^\s$ is a subexpression of $F$. E.g.
$\max(A^\s, B)^\s$ is star-chain, while
$\max(A^\s,B^\s)^\s$ is not. The expression of \Example{ex:iterexpr} is also
a star-chain, hence it is synchronised, as well as
$\min(\max(A_a,A_b)^\s,A_c)$ (note that in the latter,
$A_c$ applies on the whole input word, while $A_a$ and $A_b$ apply on
factors of it).  
\end{example}

Finitely ambiguous WA is the largest class of WA for which
emptiness, universality and comparisons are decidable \cite{DBLP:conf/fsttcs/FiliotGR14}. Already
for linearly ambiguous WA, universality and comparison problems are
undecidable \cite{DBLP:journals/corr/DaviaudGM16}. 
\Example{ex:iterexpr} is realisable by a synchronised
i-expression or a WA which non-deterministically guess, for each
factor $u_i$, whether it should count the number of $a$ or
$b$. However, as shown in~\cite{DBLP:journals/tcs/KlimannLMP04} (Section~3.5),
it is not realisable by any finitely ambiguous WA. As a consequence:

\begin{proposition}\label{prop:iexprexpr}
    There is a quantitative language $f$ such that $f$ is definable by
    a synchronised i-expression or a WA, but not by a finitely
    ambiguous WA.
\end{proposition}

As a direct consequence of the definition of i-expressions and
synchronisation, synchronised i-expressions are closed under Presburger combinators and
unambiguous iterated-sum in the following sense:

\begin{proposition}\label{prop:closureIter}
    Let $E_1,\dots,E_n,E$ be i-expressions and $\phi$ a
    functional Presburger formula of arity $n$. If
    $\text{Sync}(E_1,\dots,E_n)$, then $\phi(E_1,\dots,E_n)$ is
    synchronised, and if $E$ is synchronised, so is $E^\s$.
\end{proposition}

Despite the fact that synchronised i-expressions can express
QL that are beyond finitely ambiguous WA, we have
decidability (proved in the next section):

\begin{theorem}\label{thm:main1}
    The emptiness and universality problems are decidable for
    synchronised i-expressions. The comparisons problems for
    i-expressions $E_1,E_2$ such that $\text{Sync}\{E_1,E_2\}$ are
    decidable.
\end{theorem}

% ADD A REMARK THAT IF ONE ONLY CONSIDERS COMPARISON FOR $E_1,E_2$
% synchronised, but not mutually synchronised, then the problem is
% undecidable. 

% We conclude this section by showing that deciding whether an
% i-expression is synchronised is \textsc{Pspace-c}. 
% \begin{proposition}
    % Deciding whether any given i-expression is synchronised is
    % \textsc{Pspace-hard} and in \textsc{Exptime}, even if the WA are deterministic. 
% \end{proposition}

% \begin{proof}
    % WRONG because $\dom(A^*) = \{\epsilon\}$

        % To prove \textsc{Pspace-hardness}, we reduce the intersection
        % emptiness problem of $n$ DFA $M_1,\dots,M_n$. We let $\phi$ be
        % some arbitrary functional Presburger formula of arity $n$, and
        % $\psi$ be some arbitrary functional Presburger formula of
        % arity $2$. We also let $A$ an empty unambiguous WA (hence
        % $\dom(A) = \varnothing$), and $A_1,\dots,A_n$ be the deterministic
        % WA obtained by adding weight $0$ on all transitions of
        % $M_1,\dots,M_n$ respectively.  Then,
        % $\psi(\phi(A_1,\dots,A_n)^*, A^*)$ is synchronised iff
        % $\dom(\phi(A_1,\dots,A_n)) = \dom(A)$ iff $\bigcap_i \dom(A_i)
        % = \dom(A) = \varnothing$, iff $\bigcap_i L(M_i) =
        % \varnothing$. 

    % \end{proof}

	\section{Decidability of synchronised iterated sum expressions}\label{sec:chop}
\newcommand{\drawMainWCA}{
	\begin{tikzpicture}[>=stealth, node distance=2.8cm, thick]
		\renewcommand{\id}[1]{}
		
		% ~~~~~~~~~~~~~~~~~~~~~~~~~~~~~ STYLES ~~~~~~~~~~~~~~~~~~~~~~~~~~~~~ %
		
		\tikzstyle{every state} = [
			minimum size=.7cm
		]
		
		\tikzstyle{accepting}=[
			accepting by arrow,
			accepting text=,
			accepting where= right
		]
		
		\tikzstyle{initial}=[
			initial by arrow,
			initial text=$C$:,
			initial where= left
		]
		
		% ~~~~~~~~~~~~~~~~~~~~~~~~~~~~~ NODES ~~~~~~~~~~~~~~~~~~~~~~~~~~~~~ %
		
		\node[state, initial] (0) {\id{0}};
		\node[state] (1) [right of = 0] {\id{1}};
		\node[state, accepting] (2) [right of = 1] {\id{2}};
		
		% ~~~~~~~~~~~~~~~~~~~~~~~~~~~~~~ ARCS ~~~~~~~~~~~~~~~~~~~~~~~~~~~~~~ %
		
		\path[->]
			(0) edge node [above] {$\begin{array}{l|l} (\Sigma^*\$)^*\bullet & C_1 \end{array}$} (1)
			(1) edge node [above] {$\begin{array}{l|l} (\Sigma^*\$)^* & C_2 \end{array}$} (2)
		;
	\end{tikzpicture}
}

\newcommand{\drawSlaveWCA}[3]{
	\begin{tikzpicture}[>=stealth, node distance=2.5cm, thick]
		\renewcommand{\id}[1]{}
		
		% ~~~~~~~~~~~~~~~~~~~~~~~~~~~~~ STYLES ~~~~~~~~~~~~~~~~~~~~~~~~~~~~~ %
		
		\tikzstyle{every state} = [
			minimum size=.7cm
		]
		
		\tikzstyle{accepting}=[
			accepting by arrow,
			accepting text=,
			accepting where= right
		]
		
		\tikzstyle{initial}=[
			initial by arrow,
			initial text=$C_{#3}$:,
			initial where= left
		]
		
		% ~~~~~~~~~~~~~~~~~~~~~~~~~~~~~ NODES ~~~~~~~~~~~~~~~~~~~~~~~~~~~~~ %
		
		\node[state, initial, accepting] (0) {\id{0}};
		
		% ~~~~~~~~~~~~~~~~~~~~~~~~~~~~~~ ARCS ~~~~~~~~~~~~~~~~~~~~~~~~~~~~~~ %
		
		\path[->]
			(0) edge [loop above] node [above] {$\begin{array}{l|l} \Sigma^*\$ & \max \{ A_{#1}, A_{#2} \} \end{array}$} (0)
		;
	\end{tikzpicture}
}

In this section, we introduce a new weighted automata model, called
weighted chop automata (WCA), into which
we transform i-expressions. It is simple to see that the
proof of undecidability of i-expressions (\Theorem{thm:undec}) can
be done the same way using WCA. We introduce the class of synchronised
WCA, to which synchronised i-expressions can be compiled, and by which we recover
decidability,  thus proving \Theorem{thm:main1}. The intuitive behaviour of a WCA is as follows. An
unambiguous generalised automaton (whose transitions are not reading 
single letters but words in some regular language) ``chop'' the
input word into factors, on which expressions of the form
$\phi(C_1,\dots,C_n)$, where $C_i$ are smaller WCA, are applied to
obtain intermediate values, which are then summed to obtain the value
of the whole input word.

Formally, a \emph{generalised finite automaton} %\footnote{cf. \cite{sakarovitch}
%\textsection2.4.2 p.97}
is a tuple $A = (Q,I,F,\Delta)$ where $Q$ is a set of states, $I$ its initial states and $F$ its final states, and $\Delta$ maps any pair $(p,q)\in Q^2$ to a regular language $\Delta(p,q)\subseteq \Sigma^*$ (finitely represented by some NFA). A run of $A$ over a word $u=u_1 \dots u_n$ is a sequence $r = q_0u_1\dots q_{n-1}u_nq_{n}$ such that $u_i\in\Delta(q_{i-1},q_i)$ for all $1 \leq i \leq n$. It is accepting if $q_0\in I$ and $q_n\in F$. We say that $A$ is unambiguous if for all $u\in \Sigma^*$, there is at most one accepting run of $A$ on $u$ (and hence its decomposition $u_1\dots u_n$ is unique). This property can be decided in \textsc{PTime} (\Proposition{prop:ambiguity} in \Appendix{apx:ambiguity}).

\begin{definition} \label{def:chop}
    A $0$-weighted chop automaton is an unambiguous WA. Let $n>0$.  
	An \emph{$n$-weighted chop automaton} ($n$-WCA) is a tuple $C =
        (A, \lambda)$ where $A$ is an 
	unambiguous generalised finite automaton and $\lambda$ is a
        function mapping any pair $(p,q)\in Q^2$ to  some expression
        $E = \phi(C_1,\dots,C_m)$ where for all $i$, $C_i$ is an
        $n'$-WCA, for some $n'<n$, and $\phi$ is a functional
        Presburger formula of arity $m$. Moreover, it is required that
        at least one $C_i$ is an $(n-1)$-WCA. A WCA is an $n$-WCA for
        some $n$. 
\end{definition}

\paragraph{Semantics} A WCA $C$ defines a quantitative
language $\sem{C}$ of domain $\dom(C)$ inductively defined as
follows. If $C$ is a $0$-WCA, then its semantics is that of
unambiguous WA. Otherwise $C = (A,\lambda)$, and the set $\dom(C)$ is the set of words $u=u_1 \dots u_n$
on which there exists one accepting run $r = q_0u_1\dots
q_{n-1}u_nq_{n}$ of $A$ such that for all $1 \leq i \leq n$, if $\lambda(q_{i-1},q_i)$ is of the form
$\phi(C_1,\dots,C_m)$, then $u_i \in \bigcap_{j=1}^m \dom(C_j)$, and
in this case we let $v_i =
\sem{\phi}(\sem{C_1}(u),\dots,\sem{C_m}(u))$. The 
value of $r$ (which also defines the value of $u$) is then
$\sum_{i=1}^n v_i$.
We denote by $\text{dec}_C(u)$ the (unique) sequence
$(u_1,\lambda(q_0, q_1))\dots (u_n,\lambda(q_{n-1}, q_n))$.

\begin{example}\label{ex:wca}
Let $\Sigma = \{a,b,c,d\}$ and $\bullet,\$ \notin \Sigma$, the WCA depicted below realises the function mapping
any word of the form $u_1 \$ \dots u_n \$ \bullet v_1 \$ \dots v_m \$$,
where $u_i,v_i\in\{a,b,c,d\}^*$,
to $\sum_{i=1}^n \max(\#_a(u_i),\#_b(u_i)) + \sum_{i=1}^m
\max(\#_c(v_i),\#_d(v_i))$. The automata
$A_\sigma$ are unambiguous WA counting the number of occurences of
$\sigma$, and $C_i$ are shortcuts for $\phi_{id}(C_i)$ where
$\phi_{id}$ defines the identify function. 

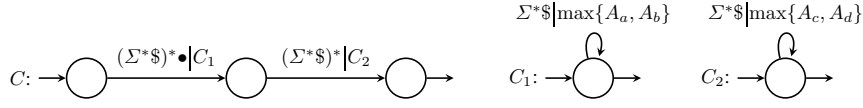
{
	\centering
	\noindent
	\begin{minipage}[b]{.55\textwidth}
		\centering
		\scalebox{.75}{\drawMainWCA}
	\end{minipage}
	\noindent
	\begin{minipage}[b]{.2\textwidth}
		\centering
		\scalebox{.75}{\drawSlaveWCA{a}{b}{1}}
	\end{minipage}
	\noindent
	\begin{minipage}[b]{.2\textwidth}
		\centering
		\scalebox{.75}{\drawSlaveWCA{c}{d}{2}}
	\end{minipage}
}
\end{example}

\paragraph{Synchronised WCA} The notion of
synchronisation of WCA is inductively defined. Two
expressions $\phi_1(C_1,\dots,C_n)$ and $\phi_2(C'_1,\dots,C'_m)$ are
synchronised if $C_i$ is synchronised with $C'_j$ for all $i,j$. 
We say that two WCA $C_1,C_2$ are
synchronised, denoted by $C_1\sync C_2$, if they are either both
$0$-WCA, or $C_1 = (A_1,\lambda_1)$ and $C_2 = (A_2,\lambda_2)$, and 
the following holds: for all $u\in L(A_1)\cap L(A_2)$, if $\text{dec}_{C_1}(u) = (u_1,E_1),\dots,(u_n,E_n)$
and $\text{dec}_{C_2}(u) = (v_1,F_1),\dots,(v_{m},F_{m})$,
then $n=m$ and for all $1 \leq i \leq n$, we have $u_i=v_i$ and $E_i$
is synchronised with $F_i$. We write
$\text{Sync}(\{C_1,\dots,C_n\})$ if $C_i \sync C_j$ for all $i, j \in
\{1, \dots, n\}$. Now, a WCA $C$ is synchronised if it is an
unambiguous WA, or it is of the form $(A,\lambda)$, and any 
expression $\phi(C_1,\dots,C_n)$ in the range of $\lambda$ satisfies
$\text{Sync}(\{C_1,\dots,C_n\})$. 
 E.g., the WCA of \Example{ex:wca} is
synchronised, and it can be seen that if $C_1\sync C_2$,
then both $C_1$ and $C_2$ are $n$-WCA for the same $n$. 
\begin{proposition} \label{prop:synchronisation}
    Synchronisation is decidable in \textsc{PTime} for WCA.
	% We can decide whether a weighted chop automaton is
	% synchronised in polynomial time.
\end{proposition}

% Note that the relation $\sync$ is not necessarily reflexive, because
% in an expression $\phi(C_1,\dots,C_m)$, the definition of
% synchronisation does not require that all the $C_i$ are
% synchronised. However, we say that an uWCA $C$ is synchronised if
% $C\sync C$. This implies, as we said, that all expressions
% $\phi(C_1,\dots,C_m)$ occurring in $C$ satisfy $C_i\sync C_j$ for all
% $i,j$. 

We now investigate the closure properties of WCA. Given two
quantitative languages $f_1,f_2$, let us define their \emph{split sum}
$f_1 \splitsum f_2$ as the function mapping any word $u$ which can be uniquely decomposed
into $u_1,u_2$ such that $u_i\in \dom(f_i)$ for all $i$, to
$f_1(u_1)+f_2(u_2)$ \cite{DBLP:conf/csl/AlurFR14}. We also define the
conditional choice $f_1\ife f_2$ as the mapping of any word
$u\in\dom(f_1)$ to $f_1(u)$, and of any word $u\in \dom(f_2)\setminus
\dom(f_1)$ to $f_2(u)$ \cite{DBLP:conf/csl/AlurFR14}. These operators
may be thought of as (unambiguous) concatenation and disjunction in
rational expressions. Synchronised WCA are closed under these
operations, as well as Presburger combinators and (unambiguous)
iterated sum, in the following sense:
\begin{proposition}\label{prop:closureChop}
    Let $C_1,\dots,C_n$ be WCA such that
    $\text{Sync}\{C_1,\dots,C_n\}$ and $C,D$ two synchronised WCA. Let
    $\phi$ be a functional Presburger formula of arity $n$, and
    $L\subseteq \Sigma^*$ a regular language.  
    There exists synchronised WCA respectively denoted by 
    $\phi(C_1,\dots,C_n)$, $C^\s$, $C\splitsum D$, $C\ife D$ and $C|_L$ such that
    \begin{itemize}
      \item for all $u\in \bigcap_{i=1}^n \dom(C_i)$, $
        \sem{\phi(C_1,\dots,C_n)}(u) =
        \sem{\phi}(\sem{C_1}(u),\dots,\sem{C_n}(u))$ and $\dom(\phi(C_1,\dots,C_n)) = \bigcap_{i=1}^n \dom(C_i)$
        
      \item $\sem{C^\s} = \sem{C}^\s$, $\sem{C\splitsum D} = \sem{C}\splitsum
        \sem{D}$, $\sem{C\ife D} = \sem{C}\ife \sem{D}$ and $\sem{C|_L} = \sem{C}|_L$
    \end{itemize}
\end{proposition}

The key lemma towards decidability of synchronised WCA is the following:

\begin{lemma} \label{lem:semilinear}
    Let $C$ be a synchronised weighted chop automaton.
    Then $\{\sem{C}(u)\mid u\in\dom(C)\}$ is semi-linear and
    effectively computable.
\end{lemma}

\begin{proof}[Sketch] The proof goes by induction on $C$. If $C$ is an
    unambiguous WA, then semi-linearity is known (for instance by
    using Parikh theorem or reversal-bounded counter machine as in the
    proof of \Theorem{thm:mono}). If $C = (A,\lambda)$ and $A$ has
    set of states $Q$, we first
    assume that for all states $p,q\in Q$, $\lambda(p,q)$ (which is an
    expression of the form $\phi(C_1,\dots,C_n)$), has semi-linear
    range $S_{p,q}$.  Consider the morphism $\mu$ from the free monoid
    $(Q\times Q)^*$ to the monoid of semi-linear sets of $\mathbb{Z}$ (with neutral
    element $\{0\}$ and addition), defined by $\mu((p,q)) =
    S_{p,q}$. Clearly, for any regular language $L\subseteq (Q\times
    Q)^*$, $\mu(L)$ is semi-linear, because semi-linear sets are
    closed under addition, finite union, and Kleene star (see \cite{Eilenberg-Schuetzenberger/69}
    for instance). Then, we can
    show that $\text{range}(C) = \mu(L)$ for $L$ the set of words over
    $Q\times Q$ of the form $(q_0,q_1)(q_1,q_2)\dots (q_k, q_{k+1})$
    such that $q_0$ is initial, $q_{k+1}$ final, and for all $i$,
    $\Delta(q_i,q_{i+1}) \neq \varnothing$. $L$ is clearly regular, as
    the $\Delta(q_i,q_{i+1})$ are.

    To show that the expressions $\phi(C_1,\dots,C_n)$ have
    semi-linear ranges, the key idea is that thanks to
    synchronisation, we can safely construct a kind of product between the
    WCA $C_1,\dots,C_n$. This product is
    not a proper WCA but a ``generalised'' WCA with values in
    $\mathbb{Z}^n$. By induction, we can show that this product has
    semi-linear range (in fact, our induction is on
    generalised WCA rather than proper WCA), whose values can be
    combined into a semilinear set thanks to the
    Presburger combinator $\phi$.\qed
\end{proof}

The following theorem is a direct consequence of
\Lemma{lem:semilinear}, \Remark{remark:zeroEmptiness} and \Proposition{prop:closureChop}.

\begin{theorem}\label{thm:decchop}
    The following problems are decidable: emptiness and universality
    of synchronised WCA, comparisons of WCA $C_1,C_2$ such that $\text{Sync}\{C_1,C_2\}$.
    % The emptiness and universality problems are decidable for
    % synchronised WCA. The comparisons problems for WCA $C_1,C_2$ s.t.
    % $\text{Sync}\{C_1,C_2\}$ is decidable.
\end{theorem}

We conclude this section by showing that any synchronised  i-expression can be converted into a
synchronised WCA. This conversion is effective, this
entails by \Theorem{thm:decchop} the decidability of synchronised
i-expressions (\Theorem{thm:main1}). 

\begin{theorem}\label{thm:iter2hybrid}
    Any synchronised i-expression $E$ is (effectively) equivalent
    to some synchronised weighted chop automaton $C_E$, i.e. $\sem{E}
    = \sem{C_E}$. 
\end{theorem}

\begin{proof}[Sketch]
    Let us illustrate the main idea of this proof on an
    example. Suppose that $E = \phi(A,B^\s)$ for some unambiguous WA
    $A,B$, and Presburger formula $\phi$. The
    difficulty with this kind of expression comes from the fact that
    $A$ is applied on the whole input word, while $B$ is applied
    iteratively on factors of it. Clearly, $A$ is also a $0$-WCA, and
    $B$ could be inductively converted into some WCA $C$, in turn used
    to construct a WCA $C^\s$ (as done in \Proposition{prop:closureChop}). However,
    $A$ and $C^\s$ are not synchronised in general: by definition of
    synchronisation for WCA, $n$-WCA are synchronised with $n$-WCA
    only. This latter property is crucial to make a product
    construction of synchronised WCA and to prove semi-linearity of
    their ranges (\Lemma{lem:semilinear}).
    
    Hence, 
    the main idea to prove this result is to ``chop'' $A$ into smaller WA that are
    synchronised with $\dom(B)$, and to express $A$ as a combination
    of these smaller automata. More precisely, for all states $p,q$
    of $A$ we can define $A_{p,q}$ to be the WA $A$ with initial state
    $p$, final state $q$, whose domain is restricted to
    $\dom(B)$. Then, all the smaller automata $A_{p,q}$ are combined
    into a single WCA which simulates successive applications of the
    automata $A_{p,q}$,
    by taking care of the fact that the words it accepts must be uniquely
    decomposable into factors of $\dom(B)$. This resulting WCA, say
    $C'$, is necessarily synchronised with $C^\s$, and we can return the single
    synchronised WCA $\phi(C',C^\s)$, as defined in
    \Proposition{prop:closureChop}, which is equivalent to
    the i-expression $\phi(A,B^\s)$. The general case is just a
    technical generalisation of this main idea. \qed
\end{proof}

% \begin{proposition}
% There is a synchronised WCA $C$ defining a quantitative language which
% is not definable by any synchronised $i$-expression. 
% \end{proposition}

	\section{Discussion}\label{sec:conclu}

% \paragraph{Summary} In this paper, we have introduced several
% formalisms for expressing quantitative languages, whose atoms are
% unambiguous (max,+)-automata, with decidable
% emptiness, universality and comparison problems. We have first considered
% monolithic expressions, which apply on the whole input word, 
% and gave tight complexity results. We have then investigated the
% extension with unambiguous sum iteration of expressions. This
% extension turned out to be undecidable, but we introduced the notion of
% synchronisation to recover decidability. Synchronised expressions with
% iterated sum are still very expressive, as they are closed under
% important operations such as Presburger combinators, unambiguous
% choice and split sum, and can express
% quantitative languages which cannot be expressed by finitely ambiguous
% (max,+)-automata, the largest class of WA for which decidability is known
% (w.r.t. the problems mentioned before). 

% \paragraph{Discussion} 
First, iterating $\max$ instead of sum also yields undecidability for
i-expressions (\Remark{rmk:undec} in Appendix). 
Second, the decidability of synchronised i-expressions goes by the model
weighted chop automata, which slice the input word into factors on
which subautomata are applied. Any synchronised i-expression can be
converted into a synchronised chop automaton
(\Theorem{thm:iter2hybrid}). We conjecture that the converse of \Theorem{thm:iter2hybrid} is
not true, i.e. synchronised WCA are strictly more expressive than synchronised
i-expressions. In particular, we conjecture that synchronised
i-expressions are not closed under split sum, unlike synchronised WCA
(\Proposition{prop:closureChop}). The quantitative language
of \Example{ex:wca} does not seem to be
definable by any synchronised i-expression.

    It turns out that extending i-expressions with split
    sum $\splitsum$ and conditional choice $\ife$, with a suitable notion of
    synchronisation, gives a formalism equivalent to synchronised
    WCA. Due to lack of space, and since the notion of
    synchronisation for such extended expressions is quite technical
    (and a bit ad-hoc), we
    decided not to include it.

    An expression formalism with unambiguous iterated sum, conditional choice and
    split sum, whose atoms are constant quantitative languages (any
    word from a regular language is mapped to a same constant value),
    was already introduced by Alur et. al.~\cite{DBLP:conf/csl/AlurFR14}. It is
    shown that this formalism is equivalent to unambiguous
	WA. Our goal was to go much beyond this
    expressivity, by having a formalism closed under Presburger
    combinators. Adding such combinators to the expressions of~\cite{DBLP:conf/csl/AlurFR14} would immediately yield an
    undecidable formalism (as a consequence of
    \Theorem{thm:undec}). This extension would actually correspond
    exactly to the extension we discussed in the previous paragraph,
    and one could come up with a notion of synchronisation by which to
    recover decidability. We did not do it in this paper, for the
    reason explained before, but it would be interesting to have an
    elegant notion of synchronisation for the extension of~\cite{DBLP:conf/csl/AlurFR14} with Presburger combinators. More
    generally, our notion of synchronisation is semantical (but
    decidable). This raises the question of whether another weighted
    expression formalism with a purely syntactic notion of
    synchronisation could be defined.

Finally, Chatterjee et. al. have introduced a recursive model of
WA~\cite{DBLP:conf/lics/ChatterjeeHO15}. They are
incomparable to weighted chop automata: they can
define QL whose ranges are not semilinear, but the
recursion depth is only 1 (a master WA calls slave WA).

\paragraph{Acknowledgements} We are very grateful to Isma\"el Jecker and
Nathan Lhote for fruitful discussions on this work, and for their help in establishing the undecidability result.

	\bibliographystyle{plain}

	\appendix
	
\section{Additional Notations}

For $w\in\Sigma^*$, $|w|$ denotes its length, $pos(w)=\{1,\dots,|w|\}$ its set of positions (in particular $pos(\epsilon)=\varnothing$), and for $i\in pos(w)$, $w[i]$ is the $i$th symbol of $w$.
A language $L$ is a subset of $\Sigma^*$.

\paragraph{Sizes of objects} We define the size $|\phi|$ of a
Presburger formula $\phi$ as the number of nodes in its syntactic
tree.

We define the representation size of a WA $M = (A,\lambda)$ with $A =
(Q,I,F,\Delta)$ as $|M| = |Q| + |\Delta|.\text{log}(\ell)$ where $\ell$ is the
maximal absolute weight of $A$.

We define the representation size $|E|$ of a m-expression $E$
inductively, if $E \equiv A$ then $|E| = |A|$ otherwise if $E \equiv
\phi(E_1,\dots,E_n)$ then $|E| = |\phi| + \sum_{i=1}^{n} |E_i|$.

\section{Proof of \Section{sec:mono}} \label{apx:mono}
	\subsection{Proof of \Lemma{lem:s-exprexpr}}

\begin{proof}
To prove that s-expressions define Lipschitz continuous functions, we need to show that for all s-expression $E$, there exists $K \in \N$ such that for all words $u,v \in \Sigma^*$:
  \begin{equation}
  	\label{eq:Lipschitz}
  	 | E(u) - E(v) | \leq K \cdot d(u,v)
  \end{equation}
Remember that $d(u,v)=|u|+|v|-2 |\sqcap(u,v)|$.

We reason by induction on the structure of the s-expressions. First, let us consider the base case where $E=D$. As $D$ is deterministic, the partial sum on $u=w \dot u'$ and $v=w \dot v'$ on their common prefix $w=\sqcap(u,v)$ is equal in the two cases to some value $s_w$ then on the two different suffixes $u'$ and $v'$, their sum may differ but at most by the following amount: $|u'| \times M + |v'| \times M$. Where $M$ is the maximum of the set of absolute value of weights appearing in the automaton $D$. It is clear that the inequality~\ref{eq:Lipschitz} is true when we take $K$=M.

Second, we consider the operation $\min$ for the inductive case, i.e. $E=\min(E_1,E_2)$. The other operators are treated similarly. By induction hypothesis, $E_1$ and $E_2$ defines Lipschitz continuous functions, and we note $K_1$ and $K_2$ their respective Lipschitz constants. We claim that $K=\max(K_1,K_2)$ is an adequate constant to show the Lipschitz continuity of $E$, i.e.:  for all words $u,v \in \Sigma^*$
   $$| \min(E_1(u) ,E_2(v)) | \leq K \cdot d(u,v)$$
Let us assume that $\min(E_1,E_2)(u)=E_1(u)$ and $\min(E_1,E_2)(v)=E_2(v)$, and that $E_1(u) \leq E_2(v)$. All the other cases are treated similarly:
  $$ \begin{array}{ll}
  	~ & | E_1(u) - E_2(v) | \\
	= & E_1(u) - E_2(v) \\
	\leq & E_1(u) - E_1(v) \\
	\leq & K_1 \cdot d(u,v) \\
	\leq & \max(K_1,K_2) \cdot d(u,v) \\
	= & K \cdot d(u,v)
 	\end{array}
	$$
	\qed
\end{proof}

\subsection{Proof of \Proposition{prop:s-exprbis}}

\begin{proof}
    We first define an unambiguous WA $A$ that realises a function, called ``last block'',  which is non Lipschitz continuous. This establishes that s-expressions are not as expressive as finite valued WA (which extends unambiguous WA). The function ``last block'',
    which associates to any word of the form $a^{n_k} b a^{n_{k-1}} b
    \dots b a^{n_0}$ where $n_0>0$, the value $n_0$, i.e. the length of the last
    block of $a$ (which necessarily exists since $n_0>0$). When
    reading the first $a$ or the first $a$ after a $b$, $A$ uses its
    non-determinism to guess whether this $a$ belongs to the last
    block or not.

    To prove the second statement, it was shown in
    \cite{DBLP:journals/tcs/KlimannLMP04} (Section~3.6) that the
    function $f:u\mapsto
    \min(\#_a(u),\#_b(u))$ is not definable by any
    WA. Clearly, $u\mapsto \#_\sigma(u)$ for
    $\sigma\in\{a,b\}$ is definable
    by a deterministic WA $A_\sigma$, hence $f$ is definable by the
    s-expression $\min(A_a,A_b)$.

    Another example is the following. Given two
    multi-sequential\footnote{Multi-sequential WA are
      finite unions of sequential, i.e. (input) deterministic, WA.}
    WA
    $B_1,B_2$ with domain $\Sigma^*$, the function $g : u\mapsto
    |B_1(u)-B_2(u)|$ is not definable by a WA,
    while it is definable by an s-expression. Multi-sequential
    automata, as they are unions of (input) deterministic WA, are
    easily defined by s-expression (by using the closure under
    $\max$). Let $S_1,S_2$ be s-expressions defining $B_1,B_2$
    respectively, then $g$ is defined by the s-expression 
    $\max(S_1-S_2,S_2-S_1)$.

    Suppose that $g$ is definable by some WA.
    Then by taking $B_1$ such that $\sem{B_1}(u) = 0$ for all
    $u$, and $B_2$ such that $\sem{B_2}(u) = \max(-\#_a(u),
    -\#_b(u))$, then $|B_1 - B_2| = f$. Clearly, $B_1$ and $B_2$ can be taken to be
    multi-sequential, and we get a contradiction because $f$ is not
    definable by any WA. \qed
\end{proof}

\subsubsection{Proof of \Lemma{lem:m-exprexpr}}

\begin{proof}
We show that m-expressions can express
any quantitative language definable by a $k$-valued WA. 
It is known that any $k$-valued WA $A$ can be decomposed into a
disjoint union of $k$ unambiguous WA $A_i$ \cite{DBLP:conf/fsttcs/FiliotGR14}. It is tempting to think
that $A$ is equivalent to the m-expression
$\max(A_1,\dots,A_k)$. However, this latter expression is
defined only on $\bigcap_i \dom(A_i)$, which may be strictly included
in $\dom(A)$. Hence, we first complete any automaton $A_i$ into $B_i$,
where $\dom(B_i) = \dom(A)$, as follows: if $\alpha$ is the smallest
value occurring on the transitions of the automata $A_i$, then, $B_i$
is the disjoint union of $A_i$ and some deterministic WA $A_i^c$ such that $\dom(A_i^c) =
\dom(A)\setminus \dom(A_i)$ and $\sem{A_i^c}(w) = \alpha |w|$. $A_i^c$
can be easily constructed from any DFA recognising $\dom(A)\setminus
\dom(A_i)$ and weight function associating $\alpha$ to any
transitions. Then, $A$ is equivalent to the m-expression
$\max(B_1,\dots,B_k)$.

For the second statement, it is already the case for s-expressions
(\Proposition{prop:s-exprbis}).\qed
\end{proof}

\subsection{Decision problems are \textsc{PSpace-c} for m-expressions}

\begin{proposition}[Normal form] \label{prop:normalForm}
	From any m-expression $E$ with $A_1, \dots, A_n$ its unambiguous weighted automata,
	one can construct in linear-time an equivalent m-expression $\phi(A_1,\dots,A_n)$,
	for some functional Presburger formula $\phi$.
\end{proposition}

\begin{proof}
	We construct $E'$ the normal form of $E$ such that $E \equiv E'$ and $|E'| = |E| + \O(1)$, by structural induction on $E$:
	
	\underline{Base case.}
	If $E = A$ then $E' = \phi_{id}(A)$ where $\phi_{id}(x,y) = (x=y)$.
 
        % we construct $\psi(x, y) = (x = y)$.
	% We define $E' = \psi(A)$.
	% We show that $E \equiv E'$.
	% $$
		% \begin{array}{ll}
			% E \equiv \psi(A) & \iff \forall u \in \dom(E) \dot \sem{E}(u) = \sem{\psi(A)}(u) \\
			% &\iff \forall u \in \dom(E) \dot \sem{\psi}(A\sem{A}(u), \sem{E}(u))\\
			% &\iff \forall u \in \dom(E) \dot  \sem{A}(u) = \sem{E}(u)
		% \end{array}
	% $$
	% Furthermore, $|E'| = |E| + \O(1)$ since $|E| = |A|$ and $|\psi(A)| = |A| + \O(1)$.
	
	\underline{Inductive step.}
	If $E = \psi(E_1, \dots, E_n)$ then we have $E_i \equiv
        \psi_i(A_{i, 1}, \dots, A_{i, m_i})$ for each $i \in [1..n]$
        by induction hypothesis. We construct $\phi$ as follow:
	$$
		\begin{array}{l}
			\phi(\dots, x_{i, 1}, \dots, x_{i, m_i}, \dots, y) =
		\\
			\qquad \exists y_1, \dots, y_n \dot \psi(y_1, \dots, y_n, y) \land \bigwedge_{i=1}^{n} \left( \psi_i(x_{i, 1}, \dots, x_{i, m_i}, y_i) \right)
		\end{array}
	$$
	We define $E' = \phi(A_{1, 1}, \dots, A_{1, m_1}, \dots, A_{n, 1}, \dots, A_{n, m_n})$.\qed
\end{proof}

\paragraph{Counter machines}
A $k$-counter machine (abbreviated as CM) is defined as a tuple $M=\langle \Sigma, X, Q, q_{\it init}, F, \Delta, \alpha, \tau, \lambda \rangle$ where
$\Sigma$ is an alphabet, $X$ is a finite set of $k$ counters interpreted over $\N$, $Q$ is a finite set of states, $q_{\it init} \in Q$ is an initial state, $F \subseteq Q$ is a set of accepting states, $\Delta : Q \times Q$ is a transition relation, $\alpha : \Delta \rightarrow \Sigma \cup \{\varepsilon\}$ assigns a letter from $\Sigma$ or the empty word $\epsilon$ to each transition, $\tau : \Delta \rightarrow \{ \texttt{=0}, \texttt{>0}, \texttt{true}\}^k$ assigns a guard to each transition, and $\lambda : \Delta \rightarrow \{\texttt{decr}, \texttt{nop}, \texttt{incr}\}^k$ assigns an update to each counter and each transition. 

A configuration of $M$ is a pair $(q,\nu)$ where $q \in Q$ is a state and $\nu : X \rightarrow \N$ is a valuation for the counters. Given a transition $\delta \in \Delta$, and two valuations $\nu,\nu' : X \rightarrow \N$, we write $\nu \models \tau(\delta)$ when the valuation $\nu$ satisfies the guard $\tau(\delta)$ (with the obvious semantics), and $(\nu,\nu') \models \lambda(\delta)$ when the update of the values of counters from $\nu$ to $\nu'$ satisfies $\lambda(\delta)$.

A computation of $M$ from $(q,\nu)$ to $(q',\nu')$ on a word $w \in \Sigma^*$ is a finite sequence of configurations and transitions $\rho=(q,\nu_0) \delta_0 (q_1,\nu_1) \delta_1 \dots \delta_{n-1} (q_n,\nu_n)$ such that $q_0=q$, $\nu_0=\nu$, $q_n=q'$, $\nu_n=\nu'$, and for all $i$, $0 \leq i < n$, $\nu_i \models \tau(\delta_i)$, $(\nu_i,\nu_{i+1}) \models \lambda(\delta_i)$, and $\alpha(\delta_0) \cdot \alpha(\delta_1) \cdot \ldots \cdot \alpha(\delta_{n-1})=w$, i.e. the concatenation of the symbols on transitions given by $\alpha$ is equal to the word $w$.

The semantics of a machine $M$ from valuation $\nu$ on a word $w$ is the set of valuations $\nu'$ such that there exists an accepting state $q' \in F$ and a computation of $M$ from $(q_{\it init},\nu)$ to $(q',\nu')$, this set of valuations 
is denoted by $\sem{M}(w, \nu)$. The language of a machine $M$ from valuation $\nu$ is the set of words $w$ such that there exists an accepting state $q' \in F$, a valuation $\nu'$, and a computation of $M$ from $(q_{\it init},\nu)$ to $(q',\nu')$ on $w$, this set of words is denoted by $L_{\nu}(M)$.  If for all words $w \in L_{\nu}(M)$, the set  $\sem{M}(w, \nu)$ is a singleton, we say that the machine $M$ is unambiguous from valuation $\nu$; in such case then $\sem{M}(w, \nu)$ can be interpreted directly as a valuation and we write for $x \in X$, $\sem{M}(w, \nu)(x)$ for the value of $x$ in the valuation reached after reading the word $w$ with $M$ from $\nu$.

While any non-trivial question about the sets $\sem{M}(w, \nu)$ and  $L_{\nu}(M)$ is undecidable for counter machines with two counters or more, Ibarra et al. has shown in~\cite{ibarra} that $\sem{M}(w, \nu)$ is semi-linear and effectively constructible when the machine $M$ is {\em reversal bounded}, and so the emptiness of $L_{\nu}(M)$ is decidable in that case. We recall the notion of reversal bounded here.  Given a computation $\rho$, a counter $x$ has $r$ reversals in $\rho$ if the counter $x$ alternates $r$ times between increment and decrement phases in $\rho$. A machine $M$ is {\em $r$-reversal bounded} if for all its computations $\rho$,  for all its counters $x \in X$, $x$ has at most has $r$ reversals in $\rho$. A machine $M$ is {\em reversal bounded} if there exists $r \in \N$ such that $M$ is $r$-reversal bounded.

\begin{theorem}[Lemma~2 of~\cite{ibarra}] \label{thm:Ibarra}
	Let $M$ be a $1$-reversal $k$-counter machine with $m$ transitions and $\nu$ a valuation of its counters.
	$L_{\nu}(M) \neq \varnothing$ if and only if $M_{\nu}$ accepts some input in $(k(m + |\nu|)^{kC})$ transitions where $|\nu| = \max \{ \nu(i) \mid 1 \leq i \leq k \}$ and $C$ is constant.
\end{theorem}

We will use counter machines as an algorithmic tool in the sequel. Some of our algorithms rely on the following lemmas that relates reversal bounded machines, Presburger expressions and weighted automata.

\begin{lemma}[Presburger term to CM] \label{term2CM}
	For all Presburger terms $t$ with $P$ the set of positions of its
        syntactic tree, and for all valuation $\nu \colon {\it
          free}(t) \rightarrow \N$, one can construct a $1$-reversal
        deterministic counter machine $M=\langle \Sigma, X, Q, q_{\it
          init}, F, \Delta, \alpha, \tau, \lambda \rangle$ (which is
        increasing-decreasing) with $X=\{ x^+_p,x^-_p \mid  p \in P\}$ and a valuation $\nu_M : X \rightarrow \N$, such that $M$ has a (unique) computation from $(q_{\it init},\nu_M)$ to $(q_f,\nu'_M)$ on the empty word $\epsilon$, and $q_f \in F$ and $\nu(p_0)=\nu'_M(x^+_{p_0})-\nu'_M(x^-_{p_0})$ where $p_0$ the position of term $t$.
	The sizes of $M$ and $\nu_M$ are bounded linearly in the sizes of $t$ and $\nu$ respectively.
\end{lemma}

\begin{proof}[Sketch]
    Let us give the main ideas on the example $t = x + x$. We want to
    construct a machine $M$ which, given a valuation of $x$, returns
    $x+x$. There are two issues: (1) counters in counter machines
    can only be positive, whereas $x$ can be negative, (2) there are
    two occurrences of $x$. To address these two issues, for each
    occurrence $p\in\{1,2\}$, we use two counters $x_p^+$ and
    $x_p^-$, such that at any moment, $x_p^+$ stores the content of
    $x$ if $x$ is positive (and $x_p^- = 0$), otherwise $x_p^-$ stores
    $-x$ (and $x_p^+ = 0$). The machine $M$ copies $x_1^+ + x_2^+$
    into $x_{p_0}^+$ ($p_0$ is the position of $t$), and $x_1^- +
    x_2^-$ into $x_{p_0}^-$. This can be done by decreasing $x_1^+$,
    $x_2^+$ and increasing $x_p^+$ (which is initially set to $0$) in
    parallel, until both counters $x_1^+$ and $x_2^+$ have reached
    $0$ (and similarly for the negative parts). 

    More generally, for every occurrence $p$ of a subterm $st$, we use
    two counters $x_{p}^+, x_{p}^-$ to encode the value of this
    subterm by $x_{p}^+ - x_{p}^-$.	We define $\nu_M$ as follows.
	Each occurrence $p_v$ of a variable $v \in \textit{free}(t)$ such that $\nu(v) \geq 0$ (resp. $< 0$) we let $\nu_M(x_{p_v}^+) = \nu(v)$ (resp. $\nu_M(x_{p_v}^-) = -\nu(v)$).
	All variables associated with a subterm which is not a free
        variable are mapped to zero by $\nu_M$.

	The construction of $M$ can be done by structural induction on $t$.
	For a constant $C \in \{0,1\}$, we can trivially construct a machine that computes a pair of counters which encodes $C$.
	For a variable occurrence, the pairs of counters corresponding
        to that occurrence is already valued by the initial valuation
        $\nu_M$ since all variables of $t$ are free.

	Finally, for an occurrence $p$ of a subterm $st$ of the form
        $t_1+t_2$ (with respective occurrences $p_1,p_2$), we
        inductively use counter machines $M_{p_1}$ and $M_{p_2}$ to
        compute the values of $x_{p_i}^+,x_{p_i}^-$ for $i=1,2$. The
        machine for $p$ is defined by the successive execution of
        $M_{p_1}$, $M_{p_2}$, and a machine which 
        realises the operations $x_{p}^+ =
        x_{p_1}^++x_{p_2}^+$ and $x_{p}^- = x_{p_1}^-+x_{p_2}^-$.

	We conclude by saying that the size of $M$ is bounded linearly
        in the sizes of $t$ since its a linear combination of constant
        size gadgets.
	\qed
\end{proof}

\begin{lemma}[Presburger formula to CM] \label{lemma:Pres2CM}
	For all functional Presburger formula $\phi$  with $P$ set of positions in its
	syntactic tree, and valuation $\nu \colon {\it free}(\phi) \rightarrow \N$ for the free variables in $\phi$, one can construct a $1$-reversal (increasing-decreasing) unambiguous counter machine $M=\langle \Sigma, X, Q, q_{\it init}, F, \Delta, \alpha, \tau, \lambda \rangle$ with $X=\{ x^+_p,x^-_p \mid p \in P\}$ and a valuation $\nu_M : X \rightarrow \N$, such that $\nu \models \phi$ if and only if $ L_{\nu_M}(M) \neq \varnothing$.
	The sizes of $M$ and $\nu_M$ are bounded linearly in the sizes of $\phi$ and $\nu$ respectively.
\end{lemma}

\begin{proof}[Sketch]
We define $\nu_M$ as in \Lemma{term2CM}. Each occurrence $p_v$ of
a free variable $v$ of $\phi$  such that $\nu(v) \geq 0$ (resp. $< 0$) we let $\nu_M(x_{p_v}^+) = \nu(v)$ (resp. $\nu_M(x_{p_v}^-) = -\nu(v)$).
All variables associated with a subterm occurring in $\phi$ which is not a free
variable are mapped to zero by $\nu_M$.

	The construction of $M$ is done by structural induction on
        $\phi$. The base case are atomic formulas, which are either of
        the form $t_1 = t_2$ or $t_1>t_2$. As done in \Lemma{term2CM},
        for all variable $x$ occurring at position $p$ in $t_i$, we
        use counters $x_p^+$ and $x_p^-$ to encode its value. We first
        construct counter machines $M_{p_1}$ and $M_{p_2}$ inductively
        corresponding to terms $t_1,t_2$, obtained by \Lemma{term2CM},
        which computes the values of $x_{p_i}^+,x_{p_i}^-$ for $i=1,2$.
	The machine for $p$ is defined by the successive execution of
        $M_{p_1}$, $M_{p_2}$, and a machine which realises the
        operations $x_{p}^+ = x_{p_1}^++x_{p_2}^-$, $x_{p}^- =
        x_{p_1}^-+x_{p_2}^+$ and then check that $x_{p}^+ > x_{p}^-$
        by decreasing simultaneously $x_{p}^+, x_{p}^-$ until $x_{p}^-
        = 0$ and $x_{p}^+>0$ (otherwise it halts and rejects the
        input).
       The procedure is similar for $\psi$ of the form $t_1
       = t_2$ (both $x_p^+$ and $x_p^-$ must reach 0 at the same time
       when decreased).

	If $\psi$ is of the form $\psi_1 \lor \psi_2$ (with respective positions $p_1,p_2$),
	the machine for $p$ execute non-derterministicaly the machines of $M_{p_1}$ or $M_{p_2}$ both inductively constructed.

	If $\psi$ is of the form $\psi_1 \land \psi_2$ (with respective positions $p_1,p_2$), we combine the machines $M_{p_1}$ and $M_{p_2}$ both inductively constructed, with some $\varepsilon$-transition from accepting state of the first one to the initial state of the second one.

	If $\psi$ is of the form $\exists x\dot \psi_1$ (with respective position $p_1$),
	we define $P_x \subseteq P$ the set of all free occurrences of $x$ in the syntactic tree of $\psi$.
	The machine for $p$ first goes to a loop which increases
        non-deterministically, and in parallel, each $x_q^+$ or
        $x_q^+$ where $q\in P_x$ and then execute the machine
        $M_{p_1}$ inductively constructed.

	We conclude by saying that the size of $M$ is bounded linearly in the sizes of $\phi$ since its almost a linear combination of constant size gadgets.
	In fact, the only gadget which have not a constant size is the gadget of quantifier case but, w.l.o.g we can assume that all non-free variables of $\phi$ is quantified once and thus the sum of all size of quantifier gadget is bounded linearly in the sizes of $\phi$.
	\qed
\end{proof}

\begin{lemma}[uWA to CM] \label{lemma:WA2CM}
	Let $A$ be an unambiguous WA with $m$ transitions and $\ell$ be the maximum among the set of absolute values of weights in $A$.
	One can construct an unambiguous $0$-reversal (increasing only) $2$-counter machine $M$ over $X = \{x^+, x^-\}$ and with $\O\left(m\ell\right)$ transitions such that $L(A) = L_{\mathbb{0}}(M)$ and $\forall u \in L(A) \dot A(u) = M(u,\mathbb{0})(x^+) - M(u,\mathbb{0})(x^-)$ where $\mathbb{0}$ maps any counter to zero.
\end{lemma}

\begin{proof}
	We construct $M$ from the same state space as $A$, same initial state and same final states, and the same transition relation.
	We handle weights as follows.
	As weights in $A$ can be positive as well as negative, and variable in $M$ can only carry nonnegative integers, we use two variables to encode the sum of weights along runs: each time a positive weights is crossed in $A$ then $x^+$ is incremented with this weight in $M$, and each time a negative weight is crossed in $A$ then $x^-$ is incremented with the opposite of this weight in $M$. As a consequence, after reading a prefix of a word in $A$ and $M$ using the same sequence of transitions, the two devices are in the same state and the difference $x^+ - x^-$ is equal to the sum of weights along the run in $A$.
	\qed
\end{proof}

\begin{proof}[of \Theorem{thm:mono}]	
	We show that the quantitative emptiness, universality and comparison problems for m-expressions are \textsc{PSpace-Complete}. 
	
	\underline{\textsc{PSpace-easiness}.}
	By \Remark{remark:zeroEmptiness} and since \textsc{PSpace} is close under negation, we must only show that the quantitative 0-emptiness problem is in \textsc{PSpace}.
	Let $E$ be a m-expression, we can assume w.l.o.g that $E$ is of the form $\phi(A_1, \dots, A_n)$ by \Proposition{prop:normalForm}.
	Then, $|E| = |\phi| + \O(1) + \sum_{i=1}^{n} |A_i|$.
	
	The idea is to reduce this problem to the emptiness of the language of a reversal bounded counter machine.
	This machine will be the concatenation of two machines $M_A$ and $M_\psi$.
	Given an input word $w \in \dom(E)$, our computation will start in $M_A$ over $w$ where all counters start with the value zero.
	$M_A$ simulates the product of all uWA i.e compute the unique valuation $\nu$ such that $\nu(i) = A_i(w)$ for each $1 \leq i \leq n$.
	Then our computation continues in $M_\psi$ over $\varepsilon$ which guesses a value $y \geq 0$ and accepts if and only if $\phi(\nu(1), \dots, \nu(n), y)$ holds.
	
	For all $1 \leq i \leq n$, let $m_i$ the number of transition in $A_i$ and $\ell_i$ the maximal of the absolute values of its weights.
	We set $m = \max \{ m_i : 1 \leq i \leq n \}$ and $\ell = \max \{ \ell_i : 1 \leq i \leq n \}$.
	For all $1 \leq i \leq n$, we apply \Lemma{lemma:WA2CM} on $A_i$ to construct the unambiguous increasing-only $2$-counter machine $M_i$ over $\{x_i^+, x_i^-\}$ with $\O(m\ell)$ transitions such that $L(A_i) = L_{\mathbb{0}}(M_i)$ and for all $u \in L(A_i)$ we have $A_i(u) = M_i(u, \mathbb{0})(x_i^+) - M(u_i, \mathbb{0})(x_i^-)$ where $\mathbb{0}$ maps any counter to zero.
	
	To ensure that all atoms read the same word, we consider the
        unambiguous increasing-only counter machine $M_A$ over
        $\bigcup_{i=1}^n \{x_i^+, x_i^-\}$ and with $\O((m\ell)^n)$
        transitions which is defined  such that  $L_{\mathbb{0}}(M_A)
        = \bigcap_{i=1}^n L(M_i)$ and for all $u \in
        L_{\mathbb{0}}(M_A)$, for all $1 \leq i \leq n$ we have
        $M_A(u, \mathbb{0})(x_i^+) = M_i(u, \mathbb{0})(x_i^+)$ and
        $M_A(u, \mathbb{0})(x_i^-) = M_i(u, \mathbb{0})(x_i^-)$. The
        machine $M_A$ is obtained as the product of the machines $M_i$
        (its set of states are tuples of states of the machines
        $M_i$). 	Now, we define the existential Presburger
        formula $\psi$ as:
	$$\psi(x_1, \dots, x_n) = \exists x \dot (x = 0 \lor x > 0)
        \land \phi(x_1, \dots, x_n, x)$$
	Let define the valuation $\nu \colon {\it free}(\psi) \rightarrow \N$ as $\nu(x_i) = M_A(u, \mathbb{0})(x_i^+) - M_A(u, \mathbb{0})(x_i^-)$ for each $1 \leq i \leq n$.
	
	By applying \Lemma{lemma:Pres2CM} on $\psi$ with valuation $\nu$, we can construct an increasing-decreasing counter machine $M_\psi$ over $\{ x^+_{u},x^-_{u} \mid  {\it u} \in \T(\psi)\}$ and a valuation $\nu_\psi : X \rightarrow \N$, such that $\nu \models \psi$ if and only if $\epsilon \in L_{\nu_\psi}(M)$.
	The sizes of $M$ and $\nu_\psi$ are bounded linear in the sizes of $\psi$ and $\nu$ respectively.
	Note that since $|\psi| = |\phi| + \O(1)$, the size of $M_\psi$ is bounded linear in the sizes of $\phi$ and then also in the size of $E$.
	
	Let consider the concatenation $M_A \cdot M_\psi$ which is the combination of the machines of $M_A$ and $M_\psi$ with some $\varepsilon$-transition from accepting states of $M_A$ to the initial state $M_\psi$.
	We show that $L_\mathbb{0}(M_A \cdot M_\psi) \neq \varnothing$
        iff there exists $w \in dom(E)$ such that $E(w) \geq 0$:
	$$
	\begin{array}{l}
		L_\mathbb{0}(M_A \cdot M_\psi) \neq \varnothing
	\\
		\iff \exists u \in L_\mathbb{0}(M_A) \dot (\bigwedge_{i=1}^n \nu(x_i) = M_A(u, \mathbb{0})(x_i^+) - M_A(u, \mathbb{0})(x_i^-)) \land L_{\nu_M}(M_\psi) \neq \varnothing
	\\
		\iff \exists u \in \bigcap_{i=1}^n L_\mathbb{0}(M_i) \dot (\bigwedge_{i=1}^n \nu(x_i) = M_i(u, \mathbb{0})(x_i^+) - M_i(u, \mathbb{0})(x_i^-)) \land \nu \models \psi
	\\
		\iff \exists u \in \bigcap_{i=1}^n L(A_i) \dot (\bigwedge_{i=1}^n \nu(x_i) = A_i(u)) \land \exists x \dot x \geq 0 \land \phi(\nu(1), \dots, \nu(n), x)
	\\
		\iff \exists u \in \dom(E) \dot \exists x \dot x \geq 0 \land \phi(A_1(u), \dots, A_n(u), x)
	\\
		\iff \exists u \in \dom(E) \dot E(u) \geq 0
	\end{array}
	$$
	\begin{algorithm}[!h]
		\DontPrintSemicolon
		\KwIn{An upper bound $B$ on the size of the minimal execution in $M_A \cdot M_\psi$,
			all atoms $A_j = (\Sigma, Q_j, q_{j, init}, F_j, \Delta_j, \gamma_j)$ for each $1 \leq j \leq n$.
		}
		\KwData{$\ell = \max \{ |\gamma_j(p, a, q)| \mid 1 \leq j \leq n \land (p, a, q)\in \Delta_i\}$}
		\KwOut{A valuation $\nu$ such that $\exists w \in \Sigma^*, \forall 1 \leq j \leq n \dot A_j(w) = \nu(j)$}
		\Begin{
			$\textit{counter} \coloneqq 0$\;
			$\nu \coloneqq (0, \dots, 0) \in \Z^n$\;
			$q \coloneqq (q_{1, init}, \dots, q_{n, init}) \in Q_1 \times \dots \times Q_n$\;
			\While{$\textit{counter} < B$}{
				guess $\textit{end} \in \{\textit{yes, no}\}$\;
				\uCase{end = yes}{
					check $q \in F_1 \times \cdots \times F_n$\;
					\KwRet{$\nu$}\;
				}
				\Case{end = no}{
					guess $a \in \Sigma$\;
					guess $q' \in Q_1 \times \dots \times Q_n$\;
					\For{$1 \leq j \leq n$}{
						check $(q(j), a, q'(j)) \in \Delta_j$\;
						$\nu(j) \coloneqq \nu(j) + \gamma_j(q(j), a, q'(j)$\;
					}
					$q \coloneqq q'$\;
					$\textit{counter} \coloneqq \textit{counter} + 1$\; 
				}
			}
		}	
		\caption{Guess a compuation of $M_{A}$ with size less than $B$ and return its final valuation $\nu$} \label{algo}
	\end{algorithm}
	In fact, $|M_A|$ is exponential in $n$ while $|E|$ is linear in $n$.
	But, by \Theorem{thm:Ibarra} we can obtain an upper bound $B$\footnote{$(\O(|\phi|)\times\O(|M_A \cdot M_\psi|))^{\O(|\phi|)}$} on the size of minimal accepting executions of $M_A \cdot M_\psi$, with $\mathbb{0}$ as initial valuation.
	Note that the binary encoding of $B$ is polynomial in $|E|$.
	Then, we can use the bound $B$ as an upper bound on the size of minimal accepting executions of $M_A$ with $\mathbb{0}$ as initial valuation.
	\Algorithm{algo} allows us to calculate a valuation $\nu$ of the atoms of $E$ with a space polynomial in $B$.
	
	Since $|M_\psi|$ is polynomial in $|\phi|$ and then in $|E|$, the machine $M_\psi$ can be effectively constructed with a space polynomial in $|E|$.
	To do that, we apply \Lemma{lemma:Pres2CM} on $\psi$ and $\nu$.
	Note that, the bound $B$ is also the upper bound on the size of minimal accepting executions of $M_\psi$ with $\nu_\psi$ as initial valuation.
	Thus, we can guess a computation smaller than $B$ of $M_{\psi}$ with $\nu_{\psi}$ as initial valuation and we check in polynomial space with $M_\psi$ if the guessed computation is correct and accepting.
	
	\underline{\textsc{PSpace-hardness}.}
	By reduction from \emph{the finite automaton intersetion problem} (INT).
	Let $A_1, \dots, A_n$ be $n$ deterministic regular automata with a common alphabet $\Sigma$, determining whether $\{ w \in \Sigma^* \mid w \in \bigcap_{i=1}^n L(A_i) \} \neq \varnothing$ is \textsc{PSpace-Hard}. %\cite{INT}\footnote{cf. Lemma~3.2.3}.
	Moreover, since \textsc{PSpace} is close under negation, we also have that $\{ w \in \Sigma^* \mid w \in \bigcap_{i=1}^n L(A_i) \} = \varnothing$ (notINT) is \textsc{PSpace-Hard}.

	For each $1 \leq i \leq n$, we construction the uWA $W_i$ such that for all $u \in \Sigma^*$ if $u \in L(A_i)$ then $W_i(u) = 1$ otherwise $W_i(u) = 0$.
	Note that for each $1 \leq i \leq n$, we can construct $W_i$ in linear time since $A_i$ is deterministic and then the regular automaton which recognise the complementary of $L(A_i)$ can be construct in linear time.
	
	We show that INT can be reduced to the quantitative emptiness problem and notINT can be reduced to the quantitative equivalence problem:
	$$
	\begin{array}{ll}
	\{ w \in \Sigma^* \mid w \in \bigcap_{i=1}^n L(A_i) \} = \varnothing &\iff \forall w \in \Sigma^*, \exists 1 \leq i \leq n \dot w \notin L(A_i) \\
	 &\iff \forall w \in \Sigma^* \dot \min(W_1(w), \dots, W_n(u)) = 0 \\
	\{ w \in \Sigma^* \mid w \in \bigcap_{i=1}^n L(A_i) \} \neq \varnothing  &\iff \exists w \in \Sigma^* \dot \min(W_1(w), \dots, W_n(u)) = 1 \\
	&\iff \exists w \in \Sigma^* \dot \min(W_1(w), \dots, W_n(u)) - 1 \geq 0
	\end{array}
	$$
	Finally, the \textsc{PSpace-Hardness} of the quantitative equivalence implies trivially the \textsc{PSpace-Hardness} of the quantitative inclusion.
	\qed
\end{proof}

\section{Proof of \Section{sec:iter}}
	\subsection{The language of uniquely decomposable words}

Let $L\subseteq \Sigma^*$ be a language. We denote by
$L^\#\subseteq L^*$ the set of words $u\in L^*$ such that either
$u=\epsilon$, or there exists at most one tuple $(u_1,\dots,u_n)\in
(L\setminus \{ \epsilon\})^n$ such that $u=u_1\dots u_n$. 

\begin{proposition}\label{prop:regdec}
    If $L\subseteq \Sigma^*$ is regular, then $L^\#$ is (effectively) regular.
\end{proposition}

\begin{proof}
Let    $A = (Q,q_0,F,\Delta)$ be a DFA recognising $L\setminus\{\epsilon\}$.  First, from $A$, one can define the
    non-deterministic automaton (with $\epsilon$-transitions) $B = (Q,\{q_0,q'_0\},F\cup
    \{q'_0\},\Delta\cup \{ (q_f,\epsilon,q_0)\mid
    q_f\in F\})$ which accepts $\epsilon$ and all words that are decomposable into
    nonempty factors of $L$. Then, by taking the product of $B$ with itself,
    and by adding a bit of memory in this product to remember whether 
    some $\epsilon$-transition was fired in parallel of a
    non-$\epsilon$ one (which implies in that case, if the two
    simulated runs of $B$ terminates, that there are two
    decompositions), one obtains an automaton which accepts all words
    which can be non-uniquely decomposed. It suffices then to
    complement this automaton, concluding the proof. \qed
\end{proof}

\subsection{The domain of an i-expression is regular: proof of \Proposition{prop:reg}}

\begin{proof}
    The domain of a WA is regular, and defined by its underlying
    finite automaton. For an expression $\phi(E_1,\dots,E_n)$, by
    induction, $\dom(E_i)$ is regular for all $i$, and by definition, 
    $\dom(\phi(E_1,\dots,E_n) = \bigcap_i \dom(E_i)$ which is regular
    since regular languages are effectively closed under
    intersection.     Consider now the case of an expression of the form $E^\s$. By
    induction hypothesis, $\dom(E)$ is regular, and since $\dom(E^\s) =
    \dom(E)^\s$, by \Proposition{prop:regdec} we get the result. \qed
    % This could be easily proven by defining
    % some MSO formula. Let $\phi_M$ is an MSO formula defining $M$, and
    % for all $\phi_M(x,y)$ holds iff the word from position $x$ to
    % position $y$ belongs to $M$ (obtained by guarding all the
    % quantifiers $\forall z$ in $\phi_M$ by $x\leq z\leq y$ and
    % $\forall Z$ by $\forall z\in Z x\leq z\leq y$)
    % then, let $\phi(X)$ be the formula
    % $$
    % \phi(X) \equiv \text{first},\text{last}\in X\forall x,y\in X\ (x < y \wedge \neg \exists z\in
    % X, x<z<y)\rightarrow \phi_M(x,y)
    % $$
    % where first,last are constants representing the first and last
    % positions of the word. This formula states that the word can be decomposed into factors that
    % belongs to $M$. Then, $\phi_{M^\#}$ defines $M^\#$:
    % $$
    % \phi_{M^\#}\equiv \phi_{\text{empty}}\vee \exists X\phi(X)\wedge \forall Y\neq X\ \neg \phi(Y)
    % $$
\end{proof}

\subsubsection{Undecidability: Proof of \Theorem{thm:undec}}\label{sec:undec}

\begin{proof}
By \Remark{remark:zeroEmptiness}, we only need to prove the undecidability of the $0$-emptiness problem to obtain the undecidability of the universality and the comparison problems. We prove this undecidability by providing a reduction from the halting problem of two-counter machines.

Let $M=\langle \Sigma, \{x,y\}, Q, q_{\it init}, F, \Delta, \tau, \lambda,\alpha \rangle$ be a deterministic two-counter machine, let $\nu_0$ be such that $\nu_0(x)=0$ and $\nu_0(y)=0$, and w.l.o.g., let us assume that the states in $F$ are the halting states of $M$.  We reduce the problem of deciding if the unique\footnote{The computation is unique as $M$ is deterministic.} computation of $M$ that starts from configuration $(q_{\it init},\nu_0)$ reaches or not an accepting state (from which it halts) to the problem of deciding if for some effectively constructible iterated-sum expression $E$, there exists a word $w \in L(E)$ such that $E(w) \geq 0$. 

Before defining $E$, we first explain how we encode computations of $M$ into words over the alphabet $\Gamma=Q \cup \{ \vdash,\dashv, \triangleright,\triangleleft, a,b \} \cup \Delta$. Let $\rho$

$$(q_0,v_0) \delta_0 (q_1,v_1) \delta_1 \dots (q_{n-1},v_{n-1}) \delta_n (q_n,v_n)$$ 

\noindent
be a computation of $M$, we encode it by the following word over $\Gamma$:

$$
	\begin{array}{l}
			\vdash q_0 a^{v_0(x)} b^{v_0(y)} \triangleleft \delta_0  \triangleright q_1 a^{v_1(x)} b^{v_1(y)} \dashv \vdash q_1 a^{v_1(x)} b^{v_1(y)}  \triangleleft \delta_1 \triangleright q_2 a^{v_2(x)} b^{v_2(y)} \dashv \dots \\
			\vdash q_{n-1} a^{v_{n-1}(x)} b^{v_{n-1}(y)} \triangleleft \delta_{n-1}  \triangleright q_n a^{v_n(x)} b^{v_n(y)} \dashv
	\end{array}
$$

So a word $w \in \Gamma^*$ encodes of the halting computation of $M$ from $\nu_0$ if the following conditions holds:
  \begin{enumerate}
  	\item the word $w$ must be in the language defined by the following regular expression $(\vdash Q  a^*  b^* \triangleleft \Delta  \triangleright Q a^*  b^* \dashv)^*$;
	\item the first element of $Q$ in $w$ is equal to $q_{\it init}$, i.e. the computation is starting in the initial state of $M$;
	\item the last element of $Q$ in $w$ belongs to set $F$, i.e. the computation reaches an accepting state of $M$;
	\item the first element of $Q$ in $w$ is directly followed by an element in $\Delta$, i.e. the computation starts from the valuation $\nu_0$;
	\item for each factor of the form $\vdash  q_1  a^{n_1}  b^{m_1}  \delta  q_2  a^{n_2}  b^{m_2} \dashv$:
		\begin{enumerate}	
			\item $\delta$ is a transition from $q_1$ to $q_2$;
			\item $(n_1,m_1) \models \tau(\delta)$, i.e. the guard of $\delta$ is satisfied;
			\item $((n_1,m_1),(n_2,m_2)) \models \lambda(\delta)$, i.e. the updates of $\delta$ are correctly realised;
		\end{enumerate}
	\item for each factor of the form $\triangleright q_1 a^{n_1} b^{m_1} \vdash \dashv  q_2 a^{n_2} b^{m_2} \triangleleft$, it is the case that:
		\begin{enumerate}
			\item $q_1=q_2$, i.e. the control state is preserved from one configuration encoding to the next one;
			\item $n_1=n_2$ and $m_1=m_2$, i.e. valuations of counters are preserved from one configuration encoding to the next one.
		\end{enumerate}
  \end{enumerate}

Let us now explain how we can construct an expression $E$ that maps a word $w$ to value $0$ if and only if this word is the encoding of an halting computation of $M$ from valuation $\nu_0$, and to a negative value otherwise.

First, we note that this can be done by providing for each of the conditions an expression which returns $0$ when the condition is satisfied and a negative value otherwise. Then we simply need to combine those expressions with the $\min$ operator: the $\min$ expression will be equal to $0$ only if all the expressions are equal to $0$, and it will be negative otherwise.

Second, we note that all the constraints in the list above are regular constraints with the exception of $5(c)$ and $6(b)$. Being regular, all the other constraints can be directly encoded as deterministic WA and so trivially as i-expressions. We concentrate here on the constraints that require the use of iteration, and we detail the construction for constraint $5(c)$ as the construction for $6(b)$ is similar and simpler.

For constraints $5(c)$, we construct an i-expression $E_{5(c)}$ that decomposes the word uniquely as factors of the form $\vdash q_1  a^{n_1}  b^{m_1}  \delta  q_2  a^{n_2}  b^{m_2} \dashv$. On each factor, we evaluate an s-expression whose value is nonnegative if and only if the update defined by $\delta$ is correctly realised in the encoding. To show how to achieve this, assume for the illustration that $\delta$ is incrementing the counter $x$ et let us show how this can be checked. The expression that we construct in this case computes the minimum of $1+n_1-n_2$ and $-1-n_1+n_2$. It should be clear that this minimum is equal to $0$ if and only if $n_2=n_1+1$ (i.e. when the increment is correctly realised). In turn, it is a simple exercise to construct a deterministic weighted automaton to compute $n_1+1-n_2$ and one to compute $n_2-n_1-1$. All the different updates can be treated similarly. Now, the i-expression $E_{5(c)}$ simply take the sum of all the values obtained locally on all the factors of the decomposition. This sum is nonnegative if and only if all the values computed locally are nonnegative. 
\qed
\end{proof}

\begin{remark}[Iteration of $\max$]\label{rmk:undec}
    Another option would be to define the semantics of $E^\s$ has an
    iteration of $\max$, i.e. $\dom(E^\s)$ is still the set of words $u$
    that are uniquely decomposed into $u_1\dots u_n$ with
    $u_i\in\dom(E)$, but $\sem{E}(u) = \max \{ \sem{E}(u_i)\mid
    i=1,\dots,n\}$. This variant is again undecidable with respect to
    emptiness, universality and comparisons.  Indeed, an careful inspection of the proof above show that in constraint $E_{5(c)}$ and $E_{6(b)}$, we can replace the iteration of sum by iteration of $\min$, or equivalently, if we first reverse the sign of all expression, by the iteration of $\max$. In that case the $\max$ will be nonpositive if and only if the two-counter machine has an halting computation.
\end{remark}

%%% Local Variables:
%%% mode: latex
%%% TeX-master: t
%%% End:

\section{Proof of \Section{sec:chop}}
	\subsection{Unambiguity of Generalised Finite Automata} \label{apx:ambiguity}

\begin{proposition} \label{prop:ambiguity}
	We can decide whether a generalised finite automaton
	is unambiguous in polynomial time.
\end{proposition}

\begin{proof}
Let $A = (Q,I,F,\Delta)$ be a generalised finite automaton whose
	languages $\Delta(p,q)$ are given by 
	NFA $A_{p,q} = (Q_{p,q}, I_{p,q}, F_{p,q}, \Delta_{p,q})$.
        We construct in polynomial time a finite transducer $T = (P, I, \{p_f\},
        \Delta')$ which, given a word $u\in L(A)$, outputs any run of
        $A$ on $u$ (seen as a word over the alphabet $Q\cup \Sigma$). Clearly, $T$ defines a function iff $A$ is
        unambiguous. We refer the reader for instance to~\cite{DBLP:journals/tcs/BealCPS03} for a definition of finite
        transducers, but let us recall that it is an automaton
        extended with outputs. In particular, the transition relation
        $\Delta'$ has type $\Delta' \subseteq P\times \Sigma^*\times
        \Gamma^*\times P$, where $\Gamma$ is the output
        alphabet. Transitions are denoted by $p\xrightarrow{u|v} q$
        where $u$ is the input word and $v$ the output word. 
        In general, a transducer defines a binary relation
        from input to output words, but functionality is decidable in \textsc{Ptime} for
        finite transducers (see for instance~\cite{DBLP:journals/tcs/BealCPS03}).

        To construct $T$, the idea is simple. We take $\Gamma = Q\cup
        \Sigma$ has output alphabet. Then, $P = \{p_f\}\cup Q \uplus
        \biguplus_{p,q\in Q} Q_{p,q}$ where $p_f$ is a new state. The
        transition function contains the following rules:
        \begin{itemize}
          \item $q\xrightarrow{\epsilon|q} s_0$ for all $q\in
            Q$ and $s_0\in \bigcup_{q'\in Q} I_{q,q'}$: when entering
            a new subautomaton $A_{q,q'}$, the transducer starts by
            writing the state $q$ on the output. 
          \item $s\xrightarrow{\sigma|\sigma} s'$ for all
            $(s,\sigma,s')\in \bigcup_{p,q\in Q}\Delta_{p,q}$: inside
            a subautomaton $A_{o,q}$, only symbols from $\Sigma$ are
            written on the output. 
          \item $s\xrightarrow{\sigma|\sigma} q$ if there are
            $p,q\in Q$ and $s'\in F_{p,q}$ such that $(s,\sigma,s')\in
            \Delta_{p,q}$: when reaching an accepting state in a
            subautomaton, the transducer make exit the subautomaton. 
          \item $q\xrightarrow{\epsilon|q} p_f$ for all $q\in F$: the
            transducer write the last seen accepting state (when the
            end of the word is reached) on the output and goes to
            $p_f$, which is accepting and is a deadlock (no outgoing
            transitions). 
        \end{itemize}
\qed
\end{proof}

\subsection{Synchronisation of Weighted Chop Automata (\Proposition{prop:synchronisation})} \label{apx:synchronisation}

First, we define the size of a WCA $C = (A,\lambda)$ where $A = (Q, I, F, \Delta)$.
It is the natural $|C|$ define as $\sum_{p, q \in Q} n_{p,q} + |Q|^2\times \ell \times \gamma \times k$ where
$n_{p,q}$ is the number of states of the NFA recognising $\Delta(p, q)$,
$\ell > 0$ is the maximal number of arguments in functionnal Presburger formula appearing in $\lambda$ plus one,
$\gamma > 0$ is the maximal size of the sub-WCA occurring in the range of $\lambda$ and
$k > 0$ is the maximal size of the Presburger formula appearing $\lambda$.

\begin{proof}[of \Proposition{prop:synchronisation}]
    We show how to check in \textsc{PTime} that two chop automata $C_1,C_2$, or two
    expressions are synchronised. 
    The algorithm is recursive. If $C_1$ and $C_2$ are both
    unambiguous WA, then the algorithm returns 1. If one of them is
    an unambiguous WA and the other not, then the algorithm returns
    0.

    Now, consider the case where we have chop automata $C_1 = (A_1,\lambda_1)$ and $C_2 = (A_2,\lambda_2)$. We first
    show how to decide the following (weaker) property:
    for all $u\in L(A_1)\cap L(A_2)$, if $\text{dec}_{C_1}(u) = (u_1,E_1),\dots,(u_n,E_n)$
    and $\text{dec}_{C_2}(u) = (v_1,F_1),\dots,(v_{m},F_{m})$,
    then $n=m$ and for all $i \in [1..n]$, we have $u_i=v_i$.

    One again the idea is to construct a transducer $T$, which defines a
    function from $\Sigma^*$ to $2^{\Sigma^*}$, whose domain is
    $L(A_1)\cap L(A_2)$. To $u\in \dom(T)$, if $\text{dec}_{C_1}(u) = (u_1,E_1),\dots,(u_n,E_n)$
    and $\text{dec}_{C_2}(u) = (v_1,F_1),\dots,(v_{m},F_{m})$, then
    $T$ returns the set of words 
    $$
    \{ a^{|u_1|}\# a^{|u_2|} \# \dots \# a^{|u_n|}, a^{|v_1|}\# a^{|v_2|} \# \dots \# a^{|v_n|}\}
    $$
    Since $u_1\dots u_n = v_1\dots v_n = u$, the latter set is a
    singleton iff the two decompositions are equal. Hence, it suffices
    to decide whether $T$ defines a function (i.e. is functional),
    which can be done in \textsc{PTime} in the size of $T$ (see for instance~\cite{DBLP:journals/tcs/BealCPS03}).

    It remains to show how to construct $T$. $T$ is the disjoint union
    of two transducers $T_1$ and $T_2$, which respectively output 
    $a^{|u_1|}\# a^{|u_2|} \# \dots \# a^{|u_n|}$ and $a^{|v_1|}\#
    a^{|v_2|} \# \dots \# a^{|v_n|}$. Consider $T_1$. It will simulate
    the behaviour of $C_1$ in the following way. If $A_1 =
    (Q,q_0,F,\Delta)$ and $(A_{p,q})_{p,q}$ are NFA that recognise
    $\Delta(p,q)$, then $T_1$ start in a copy of $A_{q_0,p}$, where
    $p$ is non-deterministically chosen. Whenever a transition
    $(\alpha,\sigma,\beta)$ of
    $A_{q_0,p}$ is fired, $T_1$ makes several choices: either $p$ is
    not final in $A_{q_0,p}$ and $T_1$ moves to state $\beta$ and write
    $a$ on the output, or $p$ is final, in that case $T_1$ may move to
    state $\beta$ while writing $a$ on the output, or move to the
    initial state of some automaton $A_{p,q}$ for some
    non-deterministically chosen state $q$, and write $a\#$ on the
    output. Then, its behaviour is the same as in $A_{q_0,p}$, but now
    in $A_{p,q}$, etc.. Clearly, $T_1$ has a polynomial size in the size of
    $C_1$.

    Now, we also want to check that $E_i\sync F_i$ for all
    subexpressions $E_i,F_i$ that occur at the same position in some
    decomposition. Formally, let $S$ be the set of expressions $E,F$
    such that there exists $u\in L(A_1)\cap L(A_2)$, such that 
    $\text{dec}_{C_1}(u) = (u_1,E_1),\dots,(u_n,E_n)$
    and $\text{dec}_{C_2}(u) = (v_1,F_1),\dots,(v_{m},F_{m})$ and
    there exists $i$ such that $E_i = E$ and $F_i = F$. We will show
    that $S$ can be computed in \textsc{PTime}. Once $S$ has been computed,
    for every pair $(\phi(C_1,\dots,C_n), \phi'(C'_1,\dots,C'_m))\in
    S$, it suffices to call the algorithm on each pair $(C_i, C'_j)$
    for all $i,j$.

    It remains to show that $S$ can be
    computed in polynomial time. Again, for all expressions $E,F$
    occurring in the range of $\lambda_1$ and $\lambda_2$
    respectively, one could define some automaton $A_{E,F}$ (of
    polynomial size) which
    accepts a word $u\in L(A_1)\cap L(A_2)$ iff $E,F$ occurs together
    in the respective decomposition of $u$, i.e. $\text{dec}_{C_1}(u)$
    and $\text{dec}_{C_1}(u)$. Then, for all these pairs, if
    $L(A_{E,F})\neq \varnothing$ (which can be checked in \textsc{PTime}), then
    we add $(E,F)$ to $S$. To construct $A_{E,F}$, the idea is to
    simulate, via a product construction,  an execution of
    $C_1$ and an execution of $C_2$ in parallel (by also simulating
    the smaller automata defining the regular languages on the
    transitions of $C_1$ and $C_2$). In this product construction, one
    bit of memory is used to remember whether a pair of states $(p_1,p_2)$ of
    $C_1$ and $C_2$ respectively, such that $E = \lambda_1(p_1)$ and
    $F = \lambda_2(p_2)$, was reached. The automaton accepts if such a
    pair was found, and the simulation of the two runs accept (meaning
    that $u\in L(A_1)\cap L(A_2)$). \qed

\end{proof}

\subsection{Closure properties of WCA: Proof of \Proposition{prop:closureChop}} %seconde 

Before proving the proposition, we need to intermediate results. 

\begin{lemma}[Unambiguous Concatenation]\label{lem:unconc}
    Given two regular languages $L_1,L_2$, if $L_1\splitsum L_2$ denotes
    the set of words $u$ which can be uniquely decomposed into
    $u_1u_2$ with $u_i\in L_i$, then there exists $2n$ regular languages
    $N_i,M_i$ such that $N_i\subseteq L_1$ and $M_i\subseteq L_2$, and
      $L_1\splitsum L_2 = \bigcup_{i=1}^n N_iM_i$.
\end{lemma}

\begin{proof}
    Let $A_1,A_2$ be two NFA recognising $L_1,L_2$ respectively. We
    can construct an automaton $A$ that accepts the words in $L_1L_2$
    which admits at least two different factorisations with respect to
    $L_1,L_2$. It suffices to simulate two runs of $A_1$ in parallel,
    whenever a copy of $A_1$ goes to an accepting state, this copy
    either stays in $A_1$ or, thanks to     an added 
    $\epsilon$-transition, goes to some initial state of
    $A_2$. We also add one bit of memory to check that
    $\epsilon$-transitions have been taken at two different moments in
    the two simulated runs. The accepting states are states
    $(q_2,q'_2,1)$ where $q_2,q'_2$ are accepting states of $A_2$. By
    complementing $A$,  we obtain an automaton, say $B$, recognising
    $L_1\splitsum L_2$.

    Now, we again make a kind of product between $B$, $A_1$ and
    $A_2$. If $Q$ are the states of $B$, $Q_1$ of $A_1$, $Q_2$ of
    $A_2$ (with initial states $I_2$), the set of states of this
    product is $Q\times (Q_1\uplus Q_2\uplus \overline{I_2})$ where
    $\overline{I_2}$ is a copy of $I_2$. $B$ initially 
    runs in parallel of $A_1$ and, when $A_1$ enters an
    accepting state, i.e. the product is in state $(q,q_1)$ where
    $q_1\in F_1$, then we add some $\epsilon$-transition 
    to any state $(q,\overline{q_2})$ where $q_2\in I_2$. Then, from
    states of this form, the product continues its simulation of $B$
    and simulates in parallel $A_2$ (in normal states $Q_2$, so that
    the copy $\overline{I_2}$ is only met once, when the product
    switches to $A_2$).     Let denote by $B\times (A_1A_2)$ this
    product. We set its accepting states to be any pair $(q,q_2)$ or $(q,\overline{q_2})$
    where $q_2$ is accepting and $q$ is accepting.

    For all states $(q,p)$ of $B\times (A_1A_2)$, we denote by
    $L_{q,p}$ and $R_{q,p}$ the left language of $(q,p)$ and the right
    language of $(q,p)$ respectively. We claim that
    $$
    L_1\splitsum L_2 = L(B) = \bigcup_{(q,\overline{q_2})\in Q\times
      \overline{I_2}} L_{q,q_2}R_{q,q_2}
    $$
    Clearly, $\bigcup_{(q,\overline{q_2})\in Q\times
      \overline{I_2}} L_{q,q_2}R_{q,q_2}\subseteq L(B)$ since the
    product also checks that the input words are accepted by $B$.
    Conversely, if $u\in L(B)$, then it is uniquely decomposed into 
    $u_1u_2$ where $u_i\in L(A_i)$. From an accepting run $r = q_1\dots q_{n+1}
    p_1\dots p_{m+1}$ on $u$ (where $n = |u_1|$ and $m = |u_2|$), an
    accepting 
    run $r_1 = \alpha_1\dots \alpha_{n+1}$ of $A_1$ on $u_1$, and a
    accepting run $r_2
    = \beta_1\dots \beta_{m+1}$ of $A_2$, we can construct the
    following accepting run of $B\times (A_1A_2)$:
    $(q_1,\alpha_1)\dots (q_{n+1},\alpha_{n+1})
    (p_1,\overline{\beta_1})(p_2,\beta_2)\dots (p_{m+1},\beta_{m+1})$.\qed
\end{proof}

\begin{proposition}[Domain regularity]\label{prop:domreg}
    The domain of any gWCA is (effectively) regular.
\end{proposition}

\begin{proof}
Let $C = (A, \lambda)$ be an $m$-gWCA where $A = (Q, I, F, \Delta)$.
We show by induction on $m$ that $\dom(C)$ is (effectively) regular.
If $m=0$ then $C$ is an unambiguous WA and its domain is (effectively) regular
(given by its underlying (input) automaton).
Otherwise, assume by induction hypothesis that for all $k<m$ the
domain of any $k$-gWCA is (effectively) regular. We construct a
generalized finite automaton to recognise $\dom(C)$. 
Let $p, q \in Q$ and $\lambda(p, q) = (v_1, \dots, v_n)$.
For all $i\in\{1,\dots,n\}$, we define $\dom(v_i) = \bigcap_{j=1}^{k}
\dom(C_j)$ if $v_i = \phi(C_1, \dots, C_k)$.
By induction hypothesis each $\dom(v_i)$ is (effectively) regular
since all $C_j$ is a $m'$-gWCA with $m' < m$.
Finally, we define $\dom(\lambda(p, q))$ as $\bigcap_{i=1}^{n}
\dom(v_i)$. By definition of $\dom(C)$, the domain of $C$ is the language of the generalized finite
automaton $A' = (Q, I, F, \Delta')$ where for all $q, p\in Q$ we have
$\Delta'(p, q) = \Delta(p, q) \cap \dom(\lambda(p, q))$.\qed
\end{proof}

\subsubsection{Proof of \Proposition{prop:closureChop}} %second

\begin{proof}
	\newcommand{\drawSplit}{
	\begin{tikzpicture}[>=stealth, node distance=3cm, thick]
		\renewcommand{\id}[1]{}
		
		% ~~~~~~~~~~~~~~~~~~~~~~~~~~~~~ STYLES ~~~~~~~~~~~~~~~~~~~~~~~~~~~~~ %
		
		\tikzstyle{every state} = [
			minimum size=.7cm
		]
		
		\tikzstyle{accepting}=[
			accepting by arrow,
			accepting text=,
			accepting where= right
		]
		
		\tikzstyle{initial}=[
			initial by arrow,
			initial text=,
			initial where= left
		]
		
		% ~~~~~~~~~~~~~~~~~~~~~~~~~~~~~ NODES ~~~~~~~~~~~~~~~~~~~~~~~~~~~~~ %
		
		\node[state, initial] (0) {\id{0}};
		\node[state] (1) [right of = 0] {\id{1}};
		\node[state, accepting] (2) [right of = 1] {\id{2}};
		
		% ~~~~~~~~~~~~~~~~~~~~~~~~~~~~~~ ARCS ~~~~~~~~~~~~~~~~~~~~~~~~~~~~~~ %
		
		\path[->]
			(0) edge node [above] {$\begin{array}{l|l} N_i & \phi_{id}(C_1) \end{array}$} (1)
			(1) edge node [above] {$\begin{array}{l|l} M_i & \phi_{id}(C_2) \end{array}$} (2)
		;
	\end{tikzpicture}
}

\newcommand{\drawIfe}{
	\begin{tikzpicture}[>=stealth, node distance=5cm, thick]
		\renewcommand{\id}[1]{}
		
		% ~~~~~~~~~~~~~~~~~~~~~~~~~~~~~ STYLES ~~~~~~~~~~~~~~~~~~~~~~~~~~~~~ %
		
		\tikzstyle{every state} = [
			minimum size=.7cm
		]
		
		\tikzstyle{accepting}=[
			accepting by arrow,
			accepting text=,
			accepting where= right
		]
		
		\tikzstyle{initial}=[
			initial by arrow,
			initial text=,
			initial where= above
		]
		
		% ~~~~~~~~~~~~~~~~~~~~~~~~~~~~~ NODES ~~~~~~~~~~~~~~~~~~~~~~~~~~~~~ %
		
		\node[state, initial] (0) {\id{0}};
		\node[state, accepting left] (1) [left of = 0] {\id{1}};
		\node[state, accepting right] (2) [right of = 0] {\id{2}};
		
		% ~~~~~~~~~~~~~~~~~~~~~~~~~~~~~~ ARCS ~~~~~~~~~~~~~~~~~~~~~~~~~~~~~~ %
		
		\path[->]
			(0) edge node [above] {$\begin{array}{l|l} \dom(C_1) & \phi_{id}(C_1) \end{array}$} (1)
			(0) edge node [above] {$\begin{array}{l|l} \dom(C_2) \setminus \dom(C_1) & \phi_{id}(C_2) \end{array}$} (2)
		;
	\end{tikzpicture}
}

\newcommand{\drawPres}{
	\begin{tikzpicture}[>=stealth, node distance=4cm, thick]
		\renewcommand{\id}[1]{}
		
		% ~~~~~~~~~~~~~~~~~~~~~~~~~~~~~ STYLES ~~~~~~~~~~~~~~~~~~~~~~~~~~~~~ %
		
		\tikzstyle{every state} = [
			minimum size=.7cm
		]
		
		\tikzstyle{accepting}=[
			accepting by arrow,
			accepting text=,
			accepting where= right
		]
		
		\tikzstyle{initial}=[
			initial by arrow,
			initial text=,
			initial where= left
		]
		
		% ~~~~~~~~~~~~~~~~~~~~~~~~~~~~~ NODES ~~~~~~~~~~~~~~~~~~~~~~~~~~~~~ %
		
		\node[state, initial] (0) {\id{0}};
		\node[state, accepting] (1) [right of = 0] {\id{1}};
		
		% ~~~~~~~~~~~~~~~~~~~~~~~~~~~~~~ ARCS ~~~~~~~~~~~~~~~~~~~~~~~~~~~~~~ %
		
		\path[->]
		(0) edge node [above] {$\begin{array}{l|l} \Sigma^* & \phi(C_1, \dots, C_n) \end{array}$} (1)
		;
	\end{tikzpicture}
}
 
	Let $C_1,\dots,C_n$ be WCA such that $\text{Sync}\{C_1,\dots,C_n\}$ and $C,D$ two synchronised WCA.
	Let $\phi$ be a functional Presburger formula of arity $n$, and $L\subseteq \Sigma^*$ a regular language.
	We show that there exists synchronised WCA respectively denoted by $\phi(C_1,\dots,C_n)$, $C^\s$, $C\splitsum D$, $C\ife D$ and $C|_L$ such that
	\begin{itemize}
	\item
		$\dom(\phi(C_1,\dots,C_n)) = \bigcap_{i=1}^n \dom(C_i)$
		and for all $u\in \bigcap_{i=1}^n \dom(C_i)$, 
		$$\sem{\phi(C_1,\dots,C_n)}(u) = \sem{\phi}(\sem{C_1}(u),\dots,\sem{C_n}(u))$$
	\item
		$\sem{C^\s} = \sem{C}^\s$,
		$\sem{C\splitsum D} = \sem{C}\splitsum\sem{D}$ and
		$\sem{C\ife D} = \sem{C}\ife \sem{D}$,
		$\sem{C|_L} = \sem{C}|_L$.
	\end{itemize}

	\underline{Closure under star.}
	$C^\s$ is the WCA $(A,\lambda)$ defined as follows.
	The only difficulty is that it should be unambiguous (because we want decomposition to be unique), hence it is not correct to add some $\epsilon$-transition from accepting states of $C$ to its initial states.
	However, it is possible to define an unambiguous NFA $B =
        (Q,I,F,\Delta)$ with a set of special states $S\subseteq Q$
        such that $L(B) = \dom(C)^\s$ and such that for all $u\in
        L(B)$, the occurrences of special states in the accepting run
        of $B$ on $u$ decomposes (uniquely) $u$ into factors that
        belong to $\dom(C)$. Then, the states of $C^*$ are the states $S\cup I$, the initial states $I$, final states $S\cap F$, and $\Delta(s,s')$ for all $s\in S\cup I$ and $s'\in S$,  is the set of words on which there is a run of $B$ from $s$ to $s'$ that does not pass by any state of $S$ (except at the end and beginning).
	This set is easily shown to be regular.
	Finally, $\lambda(s,s') = \phi_{id}(C)$ where $\phi_{id}$ defines the identity function.
      
	\underline{Closure under regular domain restriction.}
	If $C = (A,\lambda)$, then it suffices to take the product of $A$ with any DFA $B$ such that $L(B) = L$.
	Since $A$ is a generalised automaton, the product is a bit different than the usual automata product.
	Assume $A = (Q,I,F,\Delta)$ and $B = (P,p_0,F',\delta')$ (a classical DFA).
	Then $A\times B = (Q\times P, I\times \{p_0\}, F\times F', \Delta\times \delta')$ where for all $(q,p)$, $(q',p')$, $(\Delta\times \delta')((q,p),(q',p'))$ is the set of words in $\Delta(q,q')$ such that there exists a run of $B$ from state $p$ to state $p'$.
	This set is effectively regular.
	The resulting generalized NFA is unambiguous since $B$ was
        taken to be deterministic and $A$ is unambiguous.

	\underline{Closure under Presburger combinators.}
	The WCA $\phi(C_1,\dots,C_n)$ is defined as:
	\begin{center} \scalebox{.75}{\drawPres} \end{center}
	
	\underline{Closure under conditional choice.}
	The WCA $C_1\ife C_2$ is defined as
	\begin{center} \scalebox{.75}{\drawIfe} \end{center}
        Note that $\dom(C_i)$ are regular by \Proposition{prop:domreg}. 

	\underline{Closure under split sum.}
	By \Lemma{lem:unconc}, $\dom(C_1) \splitsum \dom(C_2)
        =\bigcup_{i=1}^n N_iM_i$ for some regular languages
        $N_i\subseteq \dom(C_1)$ and $M_i\subseteq \dom(C_2)$.
        For all $i=1,\dots,n$, we define the WCA $(C_1\splitsum
        C_2)|_{N_iM_i}$ as depicted below:
	\begin{center} \scalebox{.75}{\drawSplit} \end{center}
        Note that it is unambiguous since $N_iM_i\subseteq
        \dom(C_1)\splitsum \dom(C_2)$, $N_i\in\dom(C_1)$ and $M_i\in
        \dom(C_2)$. 	Futhermore, for all $u \in N_iM_i$ such that
        $u=u_1u_2$ with $u_1\in\dom(C_1)$ and $u_2\in\dom(C_2)$ we
        have $\sem{(C_1\splitsum C_2)|_{N_iM_i}}(u) = \sem{C_1}(u_1)
        +\sem{C_2}(u_2)$. Finally, we let
        $C_1\splitsum C_2 = \bigife_{i=1}^n (C_1\splitsum C_2)|_{N_iM_i}$
        (the way this expression is parenthesised, as well as the
        order in which the index $i$ are taken,  does not change the semantics). \qed
\end{proof}

\subsection{The range of a WCA is Semi-Linear (\Lemma{lem:semilinear})}

\paragraph{Intuitions}
Let us give the main intuitive ideas on how the proof works. Assume
for time being that we consider a synchronised WCA $C = (A,\lambda)$ with a single state $q$, and hence
a single transition from $q$ to $q$ on any word of
$\Delta(q,q)$. Assume that $\lambda(q,q) = \phi(C_1,\dots,C_n)$. Then,
in order to prove semilinearity of the range of $\sem{C}$, we have to
show that the set $\{ (\sem{C_1}(u),\dots,\sem{C_n}(u)\mid u\in
\Delta(q,q)\cap \bigcap \dom(C_i)\}$ is semi-linear. For a general
synchronised WCA, if one wants to prove the result by induction on its
structure, a stronger statements is needed, namely to consider tuples
of WCA instead of a single one. A stronger statement would be formally
the following: let $C_1,\dots,C_n$ be WCA such that
    $\text{Sync}\{C_1,\dots,C_n\}$. Then the set 
    $$
    \left\{ (\sem{C_1}(u),\dots,\sem{C_n}(u))\mid u\in \bigcap \dom(C_i) \right\}
    $$
    is effectively semi-linear. This statement clearly entails \Lemma{lem:semilinear} by taking
    $n=1$.  

    In the case of monolithic expressions, we had to show a similar
    result, when the $C_i$ are replaced by unambiguous WA $A_i$. This
    was solved by doing a product construction, i.e. by definition a
    single unambiguous WA $A_1\times \dots \times A_n$ with values in
    $\mathbb{Z}^n$, for which it was easily shown semi-linearity.

    The idea of the proof for WCA is similar. Thanks to
    synchronisation, we can also define a product construction of
    WCA. However, unlike the case of the product of unambiguous WA, for which
    semilinearity was shown by using reversal-bounded counter
    machines, it was not clear how to use these machines to encode the
    range of a product of WCA. The problem comes from iteration
    (induced for instance by cycles in WCA). Even if the range of some
    WCA can be defined by a reversal-bounded counter machine,
    iterating this counter machine on decompositions of  input words
    may not yield a reversal-bounded counter machine anymore.

    Therefore, our proof only rely on reversal-bounded counter machine
    for the base case (when the WCA are unambiguous WA). Then, the
    proof is inductive: we describe the decompositions of input words
    with some unambiguous regular expression, and transfer the
    operation of this regular expression into operations on
    semi-linear sets. 

\paragraph{Product construction}
    First, to formally define the product construction, we extend the
    definition of WCA to WCA with values in $\mathbb{Z}^k$. We call
    them \emph{generalised WCA}. A $0$-gWCA is a unambiguous WA with
    values in $\mathbb{Z}^k$. For $n>0$, an $n$-gWCA is defined as a WCA
    $(A,\lambda)$, except that $\lambda$ returns $k$-tuples of
    expressions of the form $E = \phi(C_1,\dots,C_n)$, where $C_i$ are
    $n$-gWCA and $\phi$ is a functional Presburger formula of arity
    $kn$ that returns a $k$-tuple of values (i.e. $\phi$ has $k(n+1)$
    free variables the last $k$ of them being the result of the
    function). The semantics and the notion of synchronisation are
    both defined the same way as for WCA. Just to make it clear for
    synchronisation, $C_1 = (A_1,\lambda_1)$ and $C_2 =
    (A_2,\lambda_2)$ are synchronised  if for all $u\in\dom(C_1)\cap \dom(C_2)$,
let set:
$$
\begin{array}{l}
	\dec{C_1}(u) = (u_1, (E_{1, 1}, \dots, E_{1, m})), \dots, (u_n, (E_{n, 1}, \dots, E_{n, m})) \\
	\dec{C_2}(u) = (u'_1, (E'_{1, 1}, \dots, E'_{1, m'})), \dots, (u'_n, (E'_{n', 1}, \dots, E'_{n', m'}))
\end{array}
$$
then $n=n'$, for each $i \in [1..n]$, $j \in [1..m]$ and $j' \in [1..m']$ we have $u_i=u'_i$ and if $E_{i, j}$ and $E'_{i, j'}$ are of the respective form $\phi(C_1,\dots,C_\ell)$ and $\phi'(C'_1,\dots,C'_{\ell'})$, then $C_k \sync C'_{k'}$ for all $k \in [1..\ell]$ and $k' \in [1..\ell']$.

\vspace{2mm}
\paragraph{Product}
If $C_1,C_2$ are unambiguous WA with values in
$\mathbb{Z}^{k_1}$ and $\mathbb{Z}^{k_2}$ respectively, then the
product construction is a classical state product construction, whose
transitions are valued in $\mathbb{Z}^{k_1+k_2}$. Given two $n$-gWCA $C_1, C_2$ ($n>0$) such that
$C_i = (Q_i, I_i, F_i, \Delta_i, \lambda_i)$ for each $i \in [1..2]$,
we define the product $C_1 \times C_2 = (Q_1 \times Q_2, I_1 \times
I_2, F_1 \times F_2, \Delta, \lambda)$ where
$\Delta((p_1,p_2),(q_1,q_2)) = \Delta_1(p_1,q_1)\cap \Delta_2(p_2,q_2)$
and $\lambda((p_1, p_2), (q_1, q_2)) = (\lambda_1(p_1, q_1),
\lambda_2(p_2, q_2))$ for all $p_1,q_1\in Q_1$ and $p_2,q_2\in Q_2$. 

\begin{lemma}\label{lem:product}
    If $\text{Sync}\{ C_1,C_2\}$, then $\dom(C_1\times C_2) =
    \dom(C_1)\cap \dom(C_2)$ and for all $u\in \dom(C_1)\cap
    \dom(C_2)$, $\sem{C_1\times C_2}(u) =
    (\sem{C_1}(u),\sem{C_2}(u))$. 
\end{lemma}
\begin{proof}
    It is a direct consequence of the construction and the definition
    of synchronisation.\qed
\end{proof}

\paragraph{Semi-linearity for generalised WCA} We come to the main
lemma of this section, which entails \Lemma{lem:semilinear}, by
taking $n=1$ and since any WCA is a gWCA.

\begin{lemma} \label{lemma:semilinear}
	Given a tuple $( C_1, \dots, C_n )$ of gWCA such that $\Sync(\{ C_1, \dots, C_n \})$
	then $\{ (\sem{C_1}(u), \dots, \sem{C_n}(u)) \mid u \in
        \bigcap_{i=1}^n L(C_i) \}$ is effectively  semi-linear.
\end{lemma}

\begin{proof}
    First, if $\Sync\{C_1,\dots,C_n\}$, then the $C_i$ are all $m$-gWCA
    for some $n$. It is because of the base case of the definition of
    synchronisation: weighted automata can be synchronised with
    weighted automata only. If $m = 0$, then the $C_i$ are all
    unambiguous WA whose transitions are valued by tuples of integers,
    and we can take their product, whose range can be shown to be
    semi-linear, just as in the case of monolithic
    expressions (see \Section{sec:mono}).

    Now, suppose that $m>0$. First, we construct the product $C = (\prod_{i=1}^{n} C_i)$ of the
    gWCA $C_i$. By \Lemma{lem:product} we have $\dom(C) =
    \bigcap_i \dom(C_i)$ and for all $u\in\dom(C)$, $\sem{C}(u) =
    (\sem{C_1}(u),\dots,\sem{C_n}(u))$. Therefore, it suffices to show
    that $\text{Range}(C) = \{ \sem{C}(u)\mid u\in \dom(C)\}$ is semi-linear to prove
    the lemma. Suppose that $C = (A,\lambda)$ where $A = (Q, I, F, \Delta)$ and $\lambda$ maps any pair $p,q\in Q$ to some $n$-ary
    tuple of expressions 
    $\lambda(p,q) =
    (\phi_1(C_1^1,\dots,C_{k_1}^1),\dots,\phi_n(C_1^1,\dots,C_{k_n}^1))$. 
    Wlog, we
    can assume that $\Delta(p,q) = \dom(C_j^i)$ for all $1\leq i\leq
    n$ and all $1\leq j\leq k_i$. 

    Otherwise, if $L = \bigcap_{1\leq i\leq n}\bigcap_{1\leq j\leq k_i} \dom(C_j^i)$ (which is
    regular by \Proposition{prop:domreg}), we restrict the domain
    of any $C^i_j$ to $L$ and replace $\Delta(p,q)$ by $\Delta(p,q)\cap
    L$, this does not change the semantics of $C$. Closure under regular domain restriction was shown for
    synchronised (non-generalised) WCA (\Proposition{prop:closureChop}), but %second
    the same proof works for synchronised gWCA. With this assumption,
    we get $\dom(C) = L(A)$. We also define 
    $$
    S_{p,q} = \text{Range}(\sem{\lambda(p,q)}) = \{
    \sem{(\phi_1(C_1^1,\dots,C_{k_1}^1)}(u),\dots,\sem{\phi_n(C_1^1,\dots,C_{k_n}^1)}(u))\mid
    u\in \Delta(p,q)\}
    $$
    Since the $C_i^j$ are $m'$-WCA for $m'<m$, by
    applying the induction hypothesis on the tuple
    $(C_1^1,\dots,C_{k_1}^1,\dots, C_1^n,\dots, C_{k_n})$ (which are
    all mutually synchronised since $C$ is synchronised), we can show semilinearity
    of $S_{p,q}$. 

    Now, sets of $n$-tuple of integers have the structure of a monoid
    $(2^{\mathbb{Z}^n}, +, \mathbb{0}_n)$ where $+$ is defined by 
    $S + S' = \{ s+s'\mid s\in S, s'\in S'\}$ and $\mathbb{0}_n = \{
    (0,\dots,0)\}$ (tuple of arity $n$). Consider the free monoid over
    $Q\times Q$ and the morphism $\mu$ from this monoid to 
    $(2^{\mathbb{Z}^n}, +, \mathbb{0}_n)$, defined by $\mu((p,q)) =
    S_{p,q}$ for all $(p,q)\in Q\times Q$. It is easily shown that for
    any language regular language $N\subseteq (Q\times Q)^*$, $\mu(N)$ is
    semi-linear (and this is effective is $N$ is given, for instance,
    by a regular expression).     It is because the $S_{p,q}$ are semi-linear, and
    semi-linear sets are (effectively) closed under sum, union, and Kleene star (see
    for instance~\cite{Eilenberg-Schuetzenberger/69}).

    Finally, consider $N \subseteq (Q\times Q)^*$ defined as the set
    of words $(p_1,p_2)(p_2,p_3)\dots (p_{k-1}, p_{k})$ such that 
    $p_1$ is initial ($p_1\in I$), $p_k\in F$, and for all $i\in
    \{1,\dots,k-1\}$, $\Delta(p_i,p_{i+1})\neq \varnothing$. We prove
    the following statements:
    \begin{enumerate}
      \item $N$ is regular
      \item $\mu(N)$ is semi-linear
      \item $\text{Range}(C) = \mu(N)$. 
    \end{enumerate}
Clearly, 2. and 3. gives the desired result, namely that
$\text{Range}(C)$ is semi-linear. Let us prove the statements:

1. Let $A_{p,q}$ be the NFA recognising $\Delta(p,q)$. It is simple to
combine the automata $A_{p,q}$ such that $\Delta(p,q)\neq\varnothing$
into a single automaton recognising $N$. For instance, one can take
the disjoint union of all $A_{p,q}$ such that $\Delta(p,q)\neq
\varnothing$, add the states of $A$, with $I$ the set of initial
states and $F$ the set of final states, and add the following
$\epsilon$-transitions: from any state $p\in Q$, add
$\epsilon$-transitions to the initial states of any NFA $A_{p,q}$, and
from any final state of any automaton $A_{p,q}$, add some
$\epsilon$-transition to $q$.

2. Since $N$ is regular, by the remark above, we get that $\mu(N)$ is
semi-linear.

3. $\subseteq$: Let $\overline{x} \in \text{Range}(C)$. Hence, there
exists $u\in \dom(C)$ such that $\sem{C}(u) =\overline{x}$. Let $p_1\xrightarrow{u_1} p_2\xrightarrow{u_2}
p_3\dots p_{k}\rightarrow{u_k} p_{k+1}$ be the accepting run of $A$
and $u$, i.e. $u_i\in\Delta(p_i,p_{i+1})$ for all
$i=1,\dots,k$. We also have $\alpha = (p_1,p_2)(p_2,p_3)\dots (p_k,p_{k+1})\in
N$. By the semantics of WCA, $\overline{x} = \overline{x_1}+\dots +
\overline{x_k}$ where $\overline{x_i} =
\sem{\lambda(p_i,p_{i+1})}(u_i)$ for all $i=1,\dots,k$. Hence
$\overline{x_i}\in S_{p_i,p_{i+1}}$ and $\overline{x}\in
\mu(\alpha)\subseteq \mu(N)$.

3. $\supseteq$: Let $\overline{x}\in \mu(N)$. By definition of $\mu$
and $N$, there exists $\alpha = (p_1,p_2)\dots (p_{k},p_{k+1})\in N$
such that $\overline{x} \in S_{p_1,p_2}+\dots + S_{p_k,p_{k+1}}$, and
hence $\overline{x} = \sum_{i=1}^k \overline{x_i}$ for some
$\overline{x_i}\in S_{p_i,p_{i+1}}$. By
definition of the $S_{p_i,p_{i+1}}$, there exists
$u_1,\dots,u_k\in\Sigma^*$ such that $u_i\in \Delta(p_i,p_{i+1})$ and 
$\overline{x_i} =
\sem{\lambda(p_i,p_{i+1})}(u_i)$. Moreover, by definition of $N$,
$p_1$ is initial and $p_{k+1}$ is final, hence $u_1\dots u_k\in
\dom(C)$, and by the semantics of $C$, $\sem{C}(u) = \sum_i
\overline{x_i}$, i.e. $\overline{x}\in \text{Range}(C)$. \qed
\end{proof}

\subsection{Proof of \Theorem{thm:iter2hybrid}}

% \begin{lemma}\label{lem:WA2WDA}
% For all tuples of iter-expressions $\mathbb{E} =
% (E_1,\dots,E_n)$ such that $\text{Sync}(\mathbb{E})$ (here
% $\mathbb{E}$ is seen as a set), and all regular languages $L\subseteq \Sigma^*$,
% one can construct a tuple of hybrid expressions $\mathbb{H} = (H_1,\dots,
% H_n)$ such that $\text{Sync}(\mathbb{H})$ and for all
% $i\in\{1,\dots,n\}$,
% $\sem{H_i} = \sem{E_i}|_L$ (in particular $\dom(H_i) = \dom(E_i)\cap L$).
% \end{lemma}

We actually prove the stronger statement below. 

\begin{lemma}\label{lem:WA2WDA}
	For all tuples of i-expressions $\mathbb{E} =
        (E_1,\dots,E_n)$ such that $\text{Sync}(\mathbb{E})$, one can construct a tuple of  WCA 
        $\mathbb{C} = (C_1,\dots,C_n)$ such that the following
        condition hold:
	\begin{enumerate}
	\item for all $i\in\{1,\dots,n\}$, $\dom(C_i) = \bigcap_{j=1}^n \dom(E_j)$,
	\item for all $i\in\{1,\dots,n\}$, for all $u\in\bigcap_{j=1}^n \dom(E_j)$, $\sem{E_i}(u) =\sem{C_i}(u)$,
	\item $\text{Sync}(\mathbb{C})$. 
	\end{enumerate}
\end{lemma}

Applied on $\mathbb{E} = (E)$, the latter lemma shows \Theorem{thm:iter2hybrid}.

\subsubsection{Notations} For an i-expression $E$, we define
$||E||\in\mathbb{N}$ inductively by  $||A|| = 0$, $||E^\s|| = ||E||+1$, $||\phi(E_1,E_2) ||=
1+||E_1||+||E_2||$. For a tuple $\mathbb{E} = (E_1,\dots,E_n)$ of
i-expressions, we let $||\mathbb{E}|| = \sum_{i=1}^n ||E_i||$.

\subsubsection{Proof of \Lemma{lem:WA2WDA}}

\begin{proof}
	The proof goes by induction on $||\mathbb{E}||$.
	If $||\mathbb{E}|| = 0$, then all expressions in $\mathbb{E}$
        are unambiguous WA $A_1,\dots,A_n$, and hence are
        $0$-WCA. However, the conditions in the lemma requires that
        the domains of all WCA are equal to $D = \bigcap_{i=1}^n
        \dom(A_i)$. It suffices to restrict $A_1,\dots,A_n$ to $D$,
        which is always possible since unambiguous WA are closed under
        regular domain restriction, and $D$ is regular. One obtains
        unambiguous WA $A'_1,\dots,A'_n$, i.e. $0$-WCA, which satisfy
        the requirements of the Lemma (for condition 3., by  the
        definition of synchronisation, both for
        i-expressions and WCA, unambiguous WA are always mutually
        synchronised). The case $||\mathbb{E}||>0$ is more involved and is a disjunction of three cases.

\underline{Case 1}
There exists $i$ such that $E_i = \phi(F_1,F_2)$
		Then we define the tuple\footnote{Note that this case
                  exhibits already the need to consider tuples of
                  expressions rather a single expression, to prove
                  \Theorem{thm:iter2hybrid} inductively.}
		$$\mathbb{E}' =
                (E_1,\dots,E_{i-1},F_1,F_2,E_{i+1},\dots,E_n) $$
                which satisfies $\text{Sync}(\mathbb{E}')$ by
                definition of synchronisation for i-expressions. 
		We also have $||\mathbb{E}'|| < ||\mathbb{E}||$, hence
                we can apply the
                induction hypothesis on $\mathbb{E}'$, and obtain a tuple 
                $$\mathbb{C}'=(C_1,\dots,C_{i-1}, C_i^1,C_i^2,
                C_{i+1},\dots,C_n)$$
                of WCA such that, if we let $D' = \bigcap_{\alpha \in
                  \mathbb{E'}} \dom(\alpha)$, we have the property $\star$:
                \begin{itemize}
                  \item for all $j\neq i$, $\dom(C_j) = \dom(C_i^1) =
                    \dom(C_i^2) = D'$,
                  \item for all $j\neq i$,
                    $\sem{E_j}|_{D'} =\sem{C_j}$, $\sem{F_1}|_{D'} =
                    \sem{C_i^1}$, $\sem{F_2}|_{D'} = \sem{C_i^2}$,
                  \item $\text{Sync}(\mathbb{C'})$. 
                \end{itemize}
                We return the tuple of WCA
		$$ \mathbb{C} = (\phi_{id}(C_1),\dots,\phi_{id}(C_{i-1}), \phi(C_i^1,C_i^2), \phi_{id}(C_{i+1}),\dots,\phi_{id}(C_n)) $$
                where $\phi_{id}$ is a Presburger-formula denoting the
                identity function, and the Preburger operations on WCA
                has been defined in \Proposition{prop:closureChop}. We
                return $\phi_{id}(C_j)$ instead of $C_j$ for all
                $j\neq i$, to preserve synchronisation. Indeed, the
                definition of synchronisation for WCA implies that all WCA in $\mathbb{C}'$
                are $k$-WCA for some $k$, and that a $k_1$-WCA and a
                $k_2$-WCA are never synchronised when $k_1\neq
                k_2$. Now, $\phi(C_i^1,C_i^2)$ is
                a $(k+1)$-WCA, as well as $\phi_{id}(C_j)$ for all
                $j\neq i$, and synchronisation is preserved,
                i.e. $\text{Sync}(\mathbb{C})$ is true.

                Let us show
                the correctness, i.e., if $D = \bigcap_j \dom(E_j)$,
                we must show that
                \begin{enumerate}
                  \item for all $C\in \mathbb{C}$, $\dom(C) = D$
                  \item for all $j\neq i$,
                    $\sem{E_j}|_D =\sem{\phi_{id}(C_j)}$, $\sem{E_i}|_D =
                    \sem{\phi(C_1^i,C_2^i)}|_D$.
                \end{enumerate}
                Since $\dom(\phi(F_1,F_2)) = \dom(F_1)\cap
                \dom(F_2)$, we have $D = D'$. Now, for all $j\neq i$,
                we have $\dom(\phi_{id}(C_j)) = \dom(C_j) = D'$ (by
                $\star$), hence $\dom(\phi_{id}(C_j)) = D$, and
                $\sem{E_j}|_D = \sem{E_j}|_{D'} = \sem{C_j}$ (by
                star), which is again equal to $\sem{\phi_{id}(C_j)}$.  
                Finally, $\dom(\phi(C_i^1,C_i^2)) = \dom(C_i^1)\cap
                \dom(C_i^2) = D'$ (by $\star$), which equals $D$, and 
                $\sem{\phi(F_1,F_2)}|_D = \{ \phi(v_1,v_2)\mid
                v_\ell\in\sem{F_i}|_D\} = \{ \phi(v_1,v_2)\mid v_\ell\in
                \sem{C_i^\ell}\} = \sem{\phi(C_i^1,C_i^2)}$.

	\underline{Case 2} $\mathbb{E} =
          (F_1^\s, \dots, F_n^\s)$. In this case we apply our induction
          hypothesis on $(F_1,\dots,F_n)$ (which is synchronised since
          $(F_1^\s, \dots, F_n^\s)$ is synchronised), obtain a tuple of WCA
          $(C_1,\dots,C_n)$, and return $(C_1^\s,
          \dots,C_n^\s)$ (the sharp operation $^\s$ on WCA has been defined
          in \Proposition{prop:closureChop}).

          Let us show the correctness.
          By synchronisation
          of $\mathbb{E}$, there is $D\subseteq \Sigma^*$ such that
          $D = \dom(F_1)=\dots=\dom(F_n)$. Let $i\in\{1,\dots,n\}$. By
          induction hypothesis, $\dom(C_i) = \bigcap_j \dom(F_j) = D$
          and $\sem{F_i}|_D = \sem{C_i}$.  
          Now, for all $i,j\in
          \{1,\dots,n\}$, 
          $$
          \dom(C_i^\s) = \dom(C_i)^\# = \dom(F_j)^\# = \dom(F_j^\s)
          $$
          We remind that $L^\#$
          is the set of words that are uniquely decomposed by factors
          in $L$. In particular, $\dom(C_i^\s) = \bigcap_j
          \dom(F_j^\s)$, satisfying condition (1) of the Lemma. Let us
          show  condition (2). Let $u\in \dom(C_i^\s)$, 
          $$
          \sem{C_i^\s}(u) = \sum_{k=1}^p \sem{C_i}(u_k) = \sum_{k=1}^p
          \sem{F_i}(u_k) = \sem{F_i^\s}(u)
          $$
          where $u_1\dots u_p$ is the unique decomposition of $u$ into
          factors in $L$. 
          Finally, to show condition (3), by induction hypothesis we
          have $\text{Sync}\{C_1,\dots,C_n\}$ and by construction of
          $C_1^\s,\dots,C_n^\s$ and since all $C_i$ have the same
          domain, synchronisation is preserved.  %first

          \underline{Case 3} The automata and star expressions
          are mixed, i.e. each $E_i$ is either an unambiguous WA or an
          iterated expression. This case is the most technical. Let us
          explain the construction with only a single automaton and a
          single star expression, i.e. $\mathbb{E} = (A, F^\s)$.
          The case where there are arbitrarily many automata and star
          expressions is more technical but not more difficult, the
          main difficulties being found in this special case. The
          difficulty comes from the fact that $A$ and $F$ does not
          apply at the same ``level'': $A$ runs on the whole input
          word, while $F$ runs on factors of it. The definition of
          synchronisation for WCA requires that any two synchronised
          WCA must ``run'' at the same level. In this sense, the
          definition of synchronisation may seem a bit strong, unlike
          that of synchronisation of i-expressions, but it makes the
          decidability (and in particular the semi-linearity property
          -- see \Lemma{lem:semilinear}) way easier to prove,
          because it allows for a product construction, just as in the
          non-iterated case (monolithic expressions).

          Hence, we have to decompose $A$ into ``partial'' weighted
          automata that will run on factors of the input word in
          $\dom(F)$.  We assume wlog that $A$ is trim\footnote{All its states are
                accessible from an initial state and co-accessible
                from a final state. It is well-known that any
                automaton can be trimmed in polynomial time.}, and for all states $p_1,p_2$ of $A$, define $A_{p_1,p_2}$ to be $A$ with only initial state $p_1$ and only accepting state $p_2$.
		Since $A$ is unambiguous and trim, so are the WA $A_{p_1,p_2}$.
 		% Wlog, we also assume that $\dom(A)\subseteq \dom(F)$
                % (it is possible since WA are closed under regular
                % domain restrictions and $\dom(F)$ is regular by
                % \Proposition{prop:domreg}).

		Let $B$ be a DFA accepting the set of words $u$ that
                are uniquely decomposed into factors in $\dom(F)$ (see
                the proof of \Proposition{prop:regdec} for a construction of $B$).
		For all states $q_1,q_2$ of $B$, let $L_{q_1,q_2}$ be the set of words in $\dom(F)$ such that there exists a run of $B$ from state $q_1$ to state $q_2$.
                Clearly:
                $$
	                \dom(F^\s) = \bigcup \left\{
						L_{q_0,q_1}L_{q_1,q_2}\dots L_{q_{m-1},q_m}
						\begin{array}{l|ll}
							&&q_0,q_1,\dots,q_m\text{ are states of $B$,}\\
							&&\text{$q_0$ is initial and $q_m$ final} \\
						\end{array}
	                \right\}
                $$
		
		Now, for all states $p_1,p_2$ of $A$ and all states
                $q_1,q_2$ of $B$, define $A_{p_1,p_2,q_1,q_2}$ as the
                WA $A_{p_1,p_2}$ restricted to the domain $L_{q_1,q_2}$.
		It can be assumed to be unambiguous as well, since $A_{p_1,p_2}$ is unambiguous.
		Then, apply the induction hypothesis on the pairs
                $(A_{p_1,p_2,q_1,q_2}, F)$ (which is synchronised) to
                get equivalent synchronised WCA $(C_{p_1,p_2,q_1,q_2}^1, C_{p_1,p_2,q_1,q_2}^2)$.
		
		We now explain how to combine these WCA into a pair of
                WCA $(C_1,C_2)$ equivalent to $(A,F^*)$ (on
                $\dom(A)\cap \dom(F^*)$). We will return the tuple $(C_1,C_2)$. Let
                $C_{p_1,p_2,q_1,q_2} = C_{p_1,p_2,q_1,q_2}^1\times
                C_{p_1,p_2,q_1,q_2}^2$, the product of this two WCA
                (hence with values in $\mathbb{Z}^2$). We now have to
                combine all these WCA into a single one $C$ running on the
                whole input word. We construct $C$ by taking the union of all
                WCA $C_{p,p',q,q'}$, and by ``merging'', for all
                states $p_1,p_2,p_3$ of $A$ and all states
                $q_1,q_2,q_3$ of $B$,  the accepting states
                of $C_{p_1,p_2,q_1,q_2}$ with the initial states of
                $C_{p_2,p_3,q_2,q_3}$. The merging operation
                is non-deterministic, because we may need to stay in
                the automaton $C_{p_1,p_2,q_1,q_2}$ even if we have
                already seen one of its accepting state: when
                $C_{p_1,p_2,q_1,q_2}$ triggers a transition to one of
                its accepting state, it either go to it, or to some
                initial state of $C_{p_2,p_3,q_2,q_3}$. The initial
                states of $C$ are the initial states of any
                $C_{p,p',q,q'}$ such that $p$ is initial in $A$, $q$
                is initial in $B$. The accepting states of $C$ are the
                accepting states of any $C_{p,p',q,q'}$ such that $p'$
                is accepting in $A$ and $q'$ is accepting in $B$. That
                way, we have $\dom(C) = \dom(A)\cap \dom(F^\s)$.

		The WCA $C_1,C_2$ are obtained from $C$ by projecting
                the pairs of expressions occurring in $C$ to their first and second
                component respectively.

		Finally, let us sketch how to proceed if there are more than one automaton $A_1,\dots,A_n$ in $\mathbb{E}$ and more than one star expression $F_1^\s,\dots,F_m^\s$ in $\mathbb{E}$.
		The idea is very similar, but we consider an automaton $B$ that accept all words that are uniquely decomposed according to $\bigcap_j \dom(F_j)$.
		Since the star expressions are synchronised, they all decompose the input word the same way, making this construction sound.
		Then, in the sub-weighted automata $A_{p,p',q,q'}$, $p$ and $p'$ are instead tuples of states of each automata $A_j$.\qed

\end{proof}


\begin{thebibliography}{10}
		
		\bibitem{DBLP:conf/csl/AlurFR14}
		R.~Alur, A.~Freilich, and M.~Raghothaman.
		\newblock Regular combinators for string transformations.
		\newblock In {\em {CSL}}, pages 9:1--9:10, 2014.
		
		\bibitem{DBLP:journals/tcs/BealCPS03}
		M.-P. B{\'{e}}al, O.~Carton, C.~Prieur, and J.~Sakarovitch.
		\newblock Squaring transducers: an efficient procedure for deciding
		functionality and sequentiality.
		\newblock {\em TCS}, 292(1), 2003.
		
		\bibitem{DBLP:conf/concur/ChatterjeeDEHR10}
		K.~Chatterjee, L.~Doyen, H.~Edelsbrunner, T.~A. Henzinger, and P.~Rannou.
		\newblock Mean-payoff automaton expressions.
		\newblock In {\em {CONCUR}}, pages 269--283, 2010.
		
		\bibitem{cdh10}
		K.~Chatterjee, L.~Doyen, and T.~A. Henzinger.
		\newblock Quantitative languages.
		\newblock {\em {ACM} Trans. Comput. Log.}, 11(4), 2010.
		
		\bibitem{DBLP:conf/lics/ChatterjeeHO15}
		K.~Chatterjee, T.~A. Henzinger, and J.~Otop.
		\newblock Nested weighted automata.
		\newblock In {\em LICS}, 2015.
		
		\bibitem{DBLP:journals/corr/DaviaudGM16}
		L.~Daviaud, P.~Guillon, and G.~Merlet.
		\newblock Comparison of max-plus automata and joint spectral radius of tropical
		matrices.
		\newblock {\em CoRR}, abs/1612.02647, 2016.
		
		\bibitem{Droste_Kuich_Vogler_2009}
		M.~Droste, W.~Kuich, and H.~Vogler.
		\newblock {\em Handbook of Weighted Automata}.
		\newblock 2009.
		
		\bibitem{Eilenberg-Schuetzenberger/69}
		S.~Eilenberg and M.~P. Sch{\"u}tzenberger.
		\newblock Rational sets in commutative monoids.
		\newblock {\em J. Algebra}, 13:173--191, 1969.
		
		\bibitem{DBLP:conf/fsttcs/FiliotGR14}
		E.~Filiot, R.~Gentilini, and J.-F. Raskin.
		\newblock Finite-valued weighted automata.
		\newblock In {\em {FSTTCS}}, pages 133--145, 2014.
		
		\bibitem{DBLP:journals/corr/abs-1111-0862}
		E.~Filiot, R.~Gentilini, and J.-F. Raskin.
		\newblock Quantitative languages defined by functional automata.
		\newblock {\em LMCS}, 11(3), 2015.
		
		\bibitem{ibarra}
		E.~M. Gurari and O.~H. Ibarra.
		\newblock The complexity of decision problems for finite-turn multicounter
		machines.
		\newblock In {\em {ICALP}}, pages 495--505.
		
		\bibitem{DBLP:journals/tcs/KlimannLMP04}
		I.~Klimann, S.~Lombardy, J.~Mairesse, and C.~Prieur.
		\newblock Deciding unambiguity and sequentiality from a finitely ambiguous
		max-plus automaton.
		\newblock {\em TCS}, 327(3), 2004.
		
		\bibitem{Krob/94}
		D.~Krob.
		\newblock The equality problem for rational series with multiplicities in the
		tropical semiring is undecidable.
		\newblock {\em Int. Jour. of Alg. and Comp.}, 4(3):405--425, 1994.
		
		\bibitem{DBLP:conf/icalp/Velner12}
		Y.~Velner.
		\newblock The complexity of mean-payoff automaton expression.
		\newblock In {\em ICALP}, 2012.
		
	\end{thebibliography}
\end{document}